\documentclass[acmsmall,nonacm]{acmart}

\usepackage{algorithm}
\usepackage[noend]{algpseudocode}
\usepackage{amsmath}
\usepackage{amsthm}
\usepackage{bm}
\usepackage{amsfonts}
\usepackage[bold,full]{complexity}
\usepackage{graphicx}
\usepackage{hyperref}
\usepackage{mathtools}
\usepackage{multicol}
\usepackage{tikz}
    \usetikzlibrary{arrows,shapes,positioning, automata}
    \usetikzlibrary{arrows,automata,positioning}
    \usetikzlibrary{decorations.markings}
    \usetikzlibrary{arrows.meta}
\usepackage[utf8]{inputenc}
\usepackage{xcolor}
\usepackage{xspace}
\usepackage[colorinlistoftodos]{todonotes}
\usepackage[shortlabels]{enumitem}
\usepackage{subcaption}

\newtheorem{thm}{Theorem}[section]
\newtheorem{claim}[thm]{Claim}

\theoremstyle{definition}

\newtheoremstyle{cited}%
{.5\baselineskip\@plus.2\baselineskip
    \@minus.2\baselineskip}
{.5\baselineskip\@plus.2\baselineskip
    \@minus.2\baselineskip}
{\itshape}
{\parindent}
{}
{.}
{.5em}
{\textsc{\thmname{#1}} \thmnote{\normalfont#3}}

\theoremstyle{cited}

\newcommand{\cost}{{\tt cost}}

\newcommand{\N}{{\ensuremath{\mathbb N}}}
\newcommand{\B}{{\ensuremath{\mathbb B}}}

\newcommand{\mycomment}[1]{}

\newcommand{\rev}[1]{#1}

\newcommand{\twomonoid}{2-monoid\xspace}
\newcommand{\twomonoids}{2-monoids\xspace}
\newcommand{\bagsetmax}{{\sc Bag-Set Maximization}\xspace}
\newcommand{\bagsetmaxdecision}{{\sc Bag-Set Maximization Decision}\xspace}

\newcommand{\true}{\mathsf{true}}
\newcommand{\false}{\mathsf{false}}

\newcommand{\ov}{\overline}

\newcommand{\card}[1]{\left\lvert #1 \right\rvert}

\newcommand{\defeq}{\stackrel{\mathrm{def}}{=}}

\newcommand{\cd}{\text{ :- }}
\newcommand{\vars}{\mathsf{vars}}
\newcommand{\atoms}{\mathsf{at}}
\newcommand{\bcq}{\textsf{BCQ}\xspace}
\newcommand{\bcqs}{\textsf{BCQs}\xspace}
\newcommand{\sjfbcq}{\textsf{SJF-BCQ}\xspace}
\newcommand{\sjfbcqs}{\textsf{SJF-BCQs}\xspace}
\newcommand{\dom}{\textsf{Dom}}

\newcommand{\supp}{\textsf{supp}}

\newcommand{\size}[1]{|#1|}

\newcommand{\shapley}{\mathsf{Shapley}}
\newcommand{\sat}{\mathsf{\#Sat}}
\newcommand{\Dexo}{\mathcal{D}^{\mathsf{x}}}
\newcommand{\Dendo}{\mathcal{D}^{\mathsf{n}}}
\newcommand{\Dr}{\mathcal{D}^{\mathsf{r}}}

\newcommand{\sharpP}{\mathsf{\#P}}
\newcommand{\np}{\mathsf{NP}}
\newcommand{\ptime}{\mathsf{P}}
\newcommand{\wone}{\mathsf{W[1]}}
\newcommand{\fpt}{\mathsf{FPT}}
\newcommand{\BCBS}{{\sf BCBS}\xspace}

\newcommand{\fptas}{{\sf FPTAS}\xspace}

\newcommand{\calD}{\mathcal{D}}
\newcommand{\calS}{\mathcal{S}}

\newcommand{\commentout}[1]{{}}

\allowdisplaybreaks

\AtBeginDocument{%
  }

\begin{document}

\title{A Unifying Algorithm for Hierarchical Queries}

\author{Mahmoud Abo Khamis}
\email{mahmoudabo@gmail.com}
\orcid{0000-0003-3894-6494}
\affiliation{%
  \institution{RelationalAI}
  \city{Berkeley}
  \state{California}
  \country{USA}
}

\author{Jesse Comer}
\email{jacomer@seas.upenn.edu}
\orcid{0009-0006-9734-3457}
\affiliation{%
  \institution{University of Pennsylvania}
  \city{Philadelphia}
  \state{Pennsylvania}
  \country{USA}}

\author{Phokion G.\ Kolaitis}
\email{kolaitis@ucsc.edu}
\orcid{0000-0002-8407-8563}
\affiliation{%
  \institution{UC Santa Cruz \& IBM Research}
  \city{Almaden}
  \state{California}
  \country{USA}
}

\author{Sudeepa Roy}
\email{sudeepa@cs.duke.edu}
\orcid{0009-0002-8300-7891}
\affiliation{%
 \institution{Duke University}
 \city{Durham}
 \state{North Carolina}
 \country{USA}}

\author{Val Tannen}
\email{val@seas.upenn.edu}
\orcid{0009-0008-6847-7274}
\affiliation{%
  \institution{University of Pennsylvania}
  \city{Philadelphia}
  \state{Pennsylvania}
  \country{USA}}

\renewcommand{\shortauthors}{Mahmoud Abo Khamis, Jesse Comer, Phokion G.\ Kolaitis, Sudeepa Roy, \& Val Tannen}

\ccsdesc[500]{Information systems~Relational database model}
\ccsdesc[500]{Theory of computation~Database query processing and optimization (theory)}
\ccsdesc[500]{Theory of computation~Data provenance}

\keywords{hierarchical queries; dichotomy; bag-set semantics; database repair; probabilistic databases; Shapley values}

\begin{abstract}
The class of hierarchical queries is known to define the boundary of the dichotomy between tractability and intractability for the following
two extensively studied problems about self-join free Boolean conjunctive queries (\sjfbcq):  (i)  evaluating a
\sjfbcq on a tuple-independent probabilistic database; (ii)  computing the Shapley value of a fact in a database on which  a \sjfbcq evaluates to true.
Here, we establish that hierarchical queries define also the boundary of the dichotomy between tractability and intractability for a different natural algorithmic problem, which we call the \emph{bag-set maximization} problem.
The bag-set maximization problem associated with a \sjfbcq $Q$ asks: given a database $\calD$, find the biggest value that $Q$ takes under bag semantics on a database $\calD'$ obtained from $\calD$ by adding at most $\theta$ facts from another given database $\calD^r$.

For non-hierarchical queries, we show
that the bag-set maximization problem is an $\np$-complete optimization problem. More significantly,
 for hierarchical queries, we show that all three aforementioned problems (probabilistic query evaluation, Shapley value computation, and bag-set maximization)  admit a single unifying polynomial-time algorithm that
operates on an abstract algebraic structure, called a   {\em \twomonoid}.
Each of the three problems requires a different instantiation of the \twomonoid~tailored for the problem at hand.
\end{abstract}

\maketitle

\section{Introduction}
\label{sec:intro}
Dichotomy theorems are classification results that pinpoint the computational complexity of every problem in a  collection of algorithmic problems of interest  by showing that some of these problems are tractable, while all others are complete for some complexity class, such as $\np$~or $\sharpP$. It is well known that the existence of a dichotomy for a collection of algorithmic problems cannot be taken for granted a priori. Indeed, Ladner \cite{DBLP:journals/jacm/Ladner75} showed that if $\ptime \not = \np$, then there are problems in $\np$ that are neither in $\ptime$ nor $\np$-complete. Thus, unless $\ptime=\np$, no dichotomy theorem for NP exists.

A Boolean conjunctive query (\bcq) $Q$ is called {\em self-join-free} if no two atoms 
in $Q$ share the same relation symbol; we will use the abbreviation
\sjfbcq~for the term \emph{self-join-free Boolean conjunctive query}. 
There is extensive  earlier work  on 
dichotomy theorems for two different collections of algorithmic problems parameterized by 
 \sjfbcqs: (i) evaluating a \sjfbcq~on a tuple-independent probabilistic database; (ii)  computing the Shapley value of a fact in a database on which a \sjfbcq~evaluates to true. 
We now describe these two problems in more detail.
\rev{Note that since all  problems studied in this paper are parameterized by a \sjfbcq $Q$, we always assume that
the query $Q$ is fixed, hence the size of the query is not used in measuring the complexity of the problem. Thus, we use {\em data complexity} as our complexity measure throughout this paper.}

{\sc   Probabilistic Query Evaluation:}  Every fixed \sjfbcq~$Q$ gives rise to the following problem: 
    Given a {\em tuple-independent} probabilistic database $\calD$ (which means that each fact in $\calD$ is viewed as an independent event with a given probability), 
    compute the {\em marginal probability} of $Q$ on $\calD$, i.e., compute the sum of the probabilities of the databases $\calD'$ that satisfy $Q$ and are obtained from $\calD$ by either keeping or removing each fact of $\calD$ independently.
    
{\sc Shapley Value Computation}: 
Every fixed \sjfbcq~$Q$ gives rise to the following problem: Given  a database  $\calD$ on which $Q$ evaluates to true and given a fact $f$ of $\calD$, compute the {\em Shapley value} of  $f$, which, informally, is the ``contribution'' of $f$  to $Q$ evaluating to true on $\calD$. Formally,
    the  Shapley value of  $f$ is defined via the following process:
    Take a permutation of the facts in $\calD$ uniformly at random, and insert the facts one by one
    while re-evaluating $Q$ after each insertion;
    the Shapley value of $f$ is the probability that the insertion of $f$ flips the answer of $Q$ from false to true.

Consider the collection of all \sjfbcqs. 
Dalvi and Suciu \cite{DalviS07} showed that {\sc Probabilistic Query Evaluation} exhibits the following dichotomy: for every fixed \sjfbcq~$Q$, either this problem is solvable in polynomial time or it is $\sharpP$-complete. More recently, Livshits, Bertossi, Kimelfeld, and Sebag \cite{DBLP:journals/lmcs/LivshitsBKS21} showed that {\sc Shapley Value Computation}  exhibits the following dichotomy: for every fixed \sjfbcq~$Q$, either this problem is solvable in polynomial time or it is $\mathsf{FP}^{\sharpP}$-complete, where $\mathsf{FP}^\sharpP$ is the class of all problems that can be solved in polynomial time using a $\sharpP$-oracle.
Remarkably, it turned out that in both these results the boundary of the dichotomy is defined by the class of \emph{hierarchical} \sjfbcqs: if $Q$ is a hierarchical \sjfbcq, then both probabilistic query evaluation and Shapley value computation are solvable in polynomial time, whereas if $Q$ is a non-hierarchical \sjfbcq, then these problems are intractable.
Recall that
a  \sjfbcq $Q$ is {\em hierarchical} if for every two variables $X$ and $Y$ occurring in $Q$,  one of the following three conditions holds:
(i) $\atoms(X)\subseteq \atoms(Y)$; (ii) $\atoms(Y)\subseteq \atoms(X)$; (iii) $\atoms(X)\cap \atoms(Y)=\emptyset$, where if $Z$ is a variable occurring in $Q$, then $\atoms(Z)$ is the set of atoms of $Q$ in which $Z$ occurs. For example, the query $Q_{\mathrm{h}}() \cd E(X,Y)\land F(Y,Z)$ is hierarchical, while the query $Q_{\mathrm{nh}}() \cd  R(X)\land S(X,Y)\land T(Y)$ is not.\footnote{We suppress existential quantifiers in our queries.}

In this paper, we introduce and investigate the collection of 
\emph{bag-set maximization} problems, which is a different collection of natural algorithmic problems parameterized by \sjfbcqs.
\rev{Before describing these problems, we briefly recall the standard notion of {\em bag-set semantic}~\cite{10.5555/260873}.
Under bag-set semantics, the input database is always a {\em set}, i.e. no duplicate facts are allowed, but the output of a query is a {\em bag}, i.e., duplicate answers are allowed.
In the special case of a Boolean conjunctive query $Q$, the value of $Q$ under bag-set semantics on a {\em set} database $\calD$, denoted $Q(\calD)$, is the number of distinct satisfying assignments of $Q$ over $\calD$.}

\rev{{\bagsetmax:} Every fixed \sjfbcq~$Q$ gives rise to the following problem: given two (set) databases $\calD$ and $\calD^r$, and a positive integer $\theta$, find the maximum value $Q(\calD')$ that $Q$ takes under bag-set semantics over all (set) databases $\calD'$ obtained from $\calD$ by adding at most $\theta$ facts from  $\calD^r$.}

Intuitively, $\theta$ is a ``budget'' and $\calD^r$ is a ``repair'' database  used to augment the set database $\calD$ with extra facts so that the value of $Q$ under bag-set semantics is maximized,
provided the ``budget'' is not exceeded.
For example, consider the following query:
\begin{align}
    Q() \cd R(A, B) \wedge S(A, C) \wedge T(A, C, D).
    \label{eq:example:query}
\end{align}
Suppose we are given the (set) database instance $\calD$ in Figure~\ref{fig:ex:bagsetmax:D}.
Initially, $Q$ has one satisfying assignment over $\calD$, namely $(A, B, C, D) = (1, 5, 2, 4)$.
Hence, $Q(\calD)=1$ under bag-set semantics.
Now suppose we are allowed to amend $\calD$ with at most $\theta = 2$ facts from the repair database $\calD^r$ in Figure~\ref{fig:ex:bagsetmax:Dr},
and we want to maximize the resulting $Q(\calD)$.
We could amend $\calD$ with the two facts $R(1, 6)$ and $R(1, 7)$ from $\calD^r$,
which would bring $Q(\calD)$ to 3 since $Q$ now has two extra satisfying assignments,
namely $(A, B, C, D) = (1, 6, 2, 4)$ and $(1, 7, 2, 4)$.
However, a better repair is to amend $\calD$ with the two facts $R(1, 6)$ and $T(1, 2, 9)$,
since this would bring $Q(\calD)$ to 4.
In this example, this would be an optimal repair, hence the answer to this \bagsetmax instance is 4.
\begin{figure}[ht!]
    \begin{subfigure}[t]{0.49\textwidth}
        \centering
        \begin{minipage}[t]{0.2\textwidth}
            \centering
            $R$~\\\vspace{-0.25cm}
            \begin{tabular}[t]{|c|c|}
                \hline
                $A$ & $B$ \\\hline
                $1$ & $5$ \\\hline
            \end{tabular}
        \end{minipage}
        \begin{minipage}[t]{0.2\textwidth}
            \centering
            $S$~\\\vspace{-0.25cm}
            \begin{tabular}[t]{|c|c|}
                \hline
                $A$ & $C$ \\\hline
                $1$ & $1$ \\
                $1$ & $2$ \\\hline
            \end{tabular}
        \end{minipage}
        \begin{minipage}[t]{0.3\textwidth}
            \centering
            $T$~\\\vspace{-0.25cm}
            \begin{tabular}[t]{|c|c|c|}
                \hline
                $A$ & $C$ & $D$\\\hline
                $1$ & $2$ & $4$\\\hline
            \end{tabular}
        \end{minipage}
        \caption{Database $\calD$}
        \label{fig:ex:bagsetmax:D}
    \end{subfigure}
    \begin{subfigure}[t]{0.49\textwidth}
        \centering
        \begin{minipage}[t]{0.2\textwidth}
            \centering
            $R$~\\\vspace{-0.25cm}
            \begin{tabular}[t]{|c|c|}
                \hline
                $A$ & $B$ \\\hline
                $1$ & $6$ \\
                $1$ & $7$ \\\hline
            \end{tabular}
        \end{minipage}
        \begin{minipage}[t]{0.2\textwidth}
            \centering
            $S$~\\\vspace{-0.25cm}
            \begin{tabular}[t]{|c|c|}
                \hline
                $A$ & $C$ \\\hline
                & \\\hline
            \end{tabular}
        \end{minipage}
        \begin{minipage}[t]{0.3\textwidth}
            \centering
            $T$~\\\vspace{-0.25cm}
            \begin{tabular}[t]{|c|c|c|}
                \hline
                $A$ & $C$ & $D$\\\hline
                $1$ & $1$ & $4$\\
                $1$ & $2$ & $9$\\\hline
            \end{tabular}
        \end{minipage}
        \caption{Database $\Dr$}
        \label{fig:ex:bagsetmax:Dr}
    \end{subfigure}
    \caption{Input instance for \bagsetmax on query $Q$ from Eq.~\eqref{eq:example:query}.
    The value of $\theta$ is 2.}
    \Description{Input instance for \bagsetmax on query $Q$ from Eq.~\eqref{eq:example:query}.
    The value of $\theta$ is 2.}
    \label{fig:ex:bagsetmax}
\end{figure}

In some sense, the {\bagsetmax} problem is  dual  to the thoroughly studied {\sc Resilience} problem \cite{DBLP:journals/pacmmod/MakhijaG23,DBLP:conf/pods/FreireGIM20,DBLP:journals/pvldb/FreireGIM15}, which asks for the minimum number of facts that can be removed from a given database so that a conjunctive query that was true on the original database becomes false in the resulting database.
It should be noted, however, that complexity results about the {\sc Resilience} problem for \sjfbcqs~do not yield complexity results about the {\bagsetmax} problem for \sjfbcqs. The reason is that one has to also ``dualize'' the query; hence, a \sjfbcq~$Q$ translates to a query $Q^*$ definable by a universal first-order sentence with some additional syntactic constraints.

On the face of it, for every fixed \sjfbcq~$Q$, {\bagsetmax} is a \emph{polynomially-bounded $\np$-optimization problem}, which means that: (i)  the optimum value is bounded by a polynomial in the size of the input; (ii) the underlying decision problem is in $\np$. Here, the underlying decision problem asks: given two databases $\calD$ and $\calD^r$, and two positive integers $\theta$ and $\tau$, is there a database $\calD'$ obtained from $\calD$ by adding at most $\theta$ facts from $\calD^r$, so that $Q(\calD')\geq \tau$?  
The polynomial boundedness of {\bagsetmax} implies that,  
for each fixed \sjfbcq~$Q$, this problem is also in the class
$\mathsf{FP}^{\np[\log(n)]}$ of function problems solvable in polynomial time using logarithmically-many calls to an $\np$-oracle, where $n$ is the size of the input instance.
As such, the  worst-case complexity of {\bagsetmax}  is lower than the worst-case complexity of {\sc Probabilistic Query Evaluation} and {\sc Shapley Value Computation}, unless $\ptime=\np$.\footnote{In fact, Krentel \cite{DBLP:journals/jcss/Krentel88} showed that if $\ptime\not = \np$, then $\mathsf{FP}^{\np[\log(n)]}\not =   \mathsf{FP}^{\np}$ (hence, also $\mathsf{FP}^{\np[\log(n)]}\not = \mathsf{FP}^\sharpP$, etc).}

We show that {\bagsetmax} exhibits the following dichotomy: if $Q$ is a hierarchical \sjfbcq, then bag-set maximization is solvable in polynomial time, whereas if $Q$ is a non-hierarchical \sjfbcq, then bag-set maximization is an $\np$-complete optimization problem, i.e., its underlying decision problem is $\np$-complete. The intractability of {\bagsetmax} for non-hierarchical queries is shown via a reduction from the {\sc Balanced Complete Bipartite Subgraph} problem \cite[Problem GT24]{DBLP:books/fm/GareyJ79} (also known as the {\sc Bipartite Clique} problem - see \cite{DBLP:journals/siamcomp/Khot06}), a problem that has also been investigated in the context of parameterized complexity by Lin \cite{DBLP:journals/jacm/Lin18}. Using the results in  \cite{DBLP:journals/jacm/Lin18}, we show that the natural parameterized versions of {\bagsetmax} are $\mathsf{W}[1]$-hard, and hence they are unlikely to be fixed-parameter tractable.


More significantly, on the tractability side, we show that for hierarchical \sjfbcqs, there is a single unifying polynomial-time\footnote{\rev{Polynomial in the input database size. Recall that we use data complexity throughout the paper.}} algorithm that solves {\bagsetmax}, {\sc Probabilistic Query Evaluation}, and {\sc Shapley Value Computation}. 
The algorithm operates on an abstract algebraic structure that we call a {\em \twomonoid}.
 A \twomonoid~resembles  a commutative semiring, except that it is {\em not} required to satisfy the distributive property. Specifically, a \twomonoid $\bm K = (K, \oplus, \otimes)$ consists of two commutative monoids $(K, \oplus)$
and $(K, \otimes)$ with neutral elements $\bm 0$ and  $\bm 1$ in $K$ for
$\oplus$ and  $\otimes$,  respectively,  
and it also satisfies the identity $\bm 0 \otimes \bm 0 = \bm 0$.
Each  of the three aforementioned problems requires a different instantiation of the \twomonoid~tailored for the problem at hand.

Interestingly, each instantiation of the \twomonoid that we consider for each of the three problems is {\em not} going to be a semiring; specifically, it will violate the distributive property.
This is expected since all the three problems are intractable for non-hierarchical
queries, including the standard non-hierarchical query $Q_{\mathrm{nh}}() \cd R(X)\land S(X,Y)\land T(Y)$,
which is acyclic.
If the distributive property were to hold, then our generic algorithm would have been able to solve all acyclic queries, including the non-hierarchical query $Q_{\mathrm{nh}}$ above, in polynomial time, which would contradict the known hardness results for these queries
in all three problems.
In some sense, the lack of the distributivity property in the three different \twomonoids considered here limits the power of our unifying algorithm to solving just the hierarchical queries, instead of the entire collection of all acyclic queries.

\paragraph*{\rev{Paper Organization}}
The rest of the paper is organized as follows:
\rev{Section~\ref{sec:overview} gives an example that illustrates the high-level idea of our unifying algorithm for hierarchical queries.}
Section~\ref{sec:prelims} contains  preliminaries.
Section~\ref{sec:bag-set-repair} contains the intractability results about the {\bagsetmax} problem for non-hierarchical queries.
\rev{Section~\ref{sec:algorithm} presents our unifying algorithm for hierarchical queries
and its instantiations for each of the three problems considered.}
Section~\ref{sec:correctness} contains the proof of correctness and the runtime analysis of this unifying algorithm. 
Finally, Section~\ref{sec:conclusion} contains concluding remarks.

\section{\rev{Example of Our Unifying Algorithm for Hierarchical Queries}}
\label{sec:overview}
\rev{In this section,} we illustrate our unifying algorithm by applying it 
to the query $Q$ from Eq.~\eqref{eq:example:query} (which is a hierarchical query) and focusing on the problems of {\sc Probabilistic Query Evaluation} and {\bagsetmax}. 
Depending on the problem at hand, we will introduce an appropriate \twomonoid $\bm K= (K, \oplus, \otimes)$
and convert an input for the problem into a {\em $\bm K$-annotated} database $\calD$, i.e., a database in which each fact $f$ is associated with a value from $K$, called the {\em annotation} of $f$.
Then, our unifying  algorithm will consist of  a sequence of operations over $\bm K$-annotated relations.

Suppose first  that 
we are trying to solve the 
{\sc Probabilistic Query Evaluation} problem for  $Q$.
In this case, the input to the problem is a tuple-independent probabilistic database $\calD$ where each fact $f$ is associated with its probability $p_f\in [0, 1]$.
We define a \twomonoid $\bm K=(K, \oplus, \otimes)$ where $K=[0,1]$ is the domain of probabilities,  and $\otimes$ and $\oplus$ are defined as follows for all $p_1, p_2 \in [0, 1]$:
\begin{align}
    p_1 \otimes p_2 &\defeq p_1 \times p_2\label{eq:prob-semiring:otimes}\\
    p_1 \oplus p_2 &\defeq 1-(1-p_1)\times(1-p_2)= p_1 + p_2 - p_1\times p_2\label{eq:prob-semiring:oplus}
\end{align}
Note that $p_1\otimes p_2$ is the probability of the {\em conjunction} of two independent events
with probabilities $p_1$ and $p_2$, whereas $p_1\oplus p_2$ is the probability of their {\em disjunction}.
Moreover, note that $\otimes$ does {\em not} distribute over $\oplus$ above, i.e., $p_1\otimes(p_2 \oplus p_3)$ does not equal $(p_1\otimes p_2)\oplus(p_1\otimes p_3)$ in general.
The input probabilistic database  can be viewed as a $\bm K$-annotated database instance where each fact $f$ is annotated with its probability that $f$ appears when an instance (a possible world) is sampled.
The output  is a value in $K$ representing the probability of $Q$ evaluating to true on a random instance.
Our unified algorithm computes this output using a sequence of operations over $\bm K$-annotated relations that we depict below. 
For this particular application, our unified algorithm specializes precisely
to the algorithm of Dalvi and Suciu \cite{DalviS07} for evaluating a \sjfbcq~on a tuple-independent probabilistic database.
Consider again the query in Eq.~\eqref{eq:example:query}. Below, we use capital letters $A, B$ to denote variables 
and small letters $a, b$ to denote values for the corresponding variables, where the values come from a domain $\dom$.
Given a tuple $(a, c, d) \in \dom^3$,
we use $T(a, c, d)$ to denote the {\em $\bm K$-annotation} of the tuple $(a, c, d)$ in the $\bm K$-annotated relation $T$, which, in this case, corresponds to the probability of the tuple $(a, c, d)$ appearing in the relation $T$:
\begin{align}
    T'(a, c) &\gets \bigoplus_{d \in \dom} T(a, c, d), &\forall (a, c) \in \dom^2\label{eq:example:T'}\\
    S'(a, c) &\gets S(a, c)\otimes T'(a, c), &\forall (a, c) \in \dom^2\label{eq:example:S'}\\
    S''(a) &\gets \bigoplus_{c \in \dom}S'(a, c), &\forall a \in \dom\label{eq:example:S''}\\
    R'(a) &\gets \bigoplus_{b \in \dom}R(a, b), &\forall a \in \dom\label{eq:example:R'}\\
    R''(a) &\gets R'(a)\otimes S''(a), &\forall a \in \dom\label{eq:example:R''}\\
    Q() &\gets \bigoplus_{a \in \dom}R''(a)\label{eq:example:Q}
\end{align}
Consider the first operation in Eq.~\eqref{eq:example:T'}.
Given a tuple $(a, c)\in\dom^2$, this rule computes $T'(a, c)$, which is the probability of a fact of the form $T(a,c,d)$ appearing in an instance for some value of $d$.
Note that the facts $T(a, c, d)$ for different $d$ values are independent events, hence the probability of their disjunction
can be correctly computed using the $\oplus$ operator from Eq.~\eqref{eq:prob-semiring:oplus}.
In contrast, the second operation in Eq.~\eqref{eq:example:S'} computes the probability of $(a, c)$
occurring in both $S$ and $T'$, which are two independent events.
Hence, the probability of their conjunction can be computed using the $\otimes$ operator from Eq.~\eqref{eq:prob-semiring:otimes}.
The last rule in Eq.~\eqref{eq:example:Q} computes the probability of $Q$ evaluating to true over possible databases.

Next,  suppose that 
we are trying to solve
the {\bagsetmax} problem for $Q$.
In this paper, we establish that we can still use the exact same algorithm above, depicted in Eq.~\eqref{eq:example:T'}-\eqref{eq:example:Q}.
The only thing that will change is how we define the \twomonoid $\bm K=(K, \oplus, \otimes)$
and how we construct a $\bm K$-annotated database to use as input to the above algorithm.
Consider an input instance $(\calD, \Dr, \theta)$ to 
the \bagsetmax problem.
For each fact $f$, we will need to store a (finitely representable) vector $\bm x$ of natural numbers, i.e., $\bm x \in \N^\N$, 
where the $i$-the entry of this vector represents the maximum multiplicity of the fact $f$ that can be achieved with a repaired budget of $i$, that is, by adding up to $i$ facts to the database $\calD$ we are trying to repair.
In particular, the domain $K$ of the \twomonoid will be the set of all vectors $\bm x \in \N^\N$.
Given two vectors $\bm x_1, \bm x_2 \in \N^\N$, we define the $\oplus$ and $\otimes$ operators as follows:
\begin{align}
    (\bm x_1 \oplus \bm x_2)(i) &\defeq \max_{i_1, i_2 \in \N: i_1 + i_2 = i} \bm x_1(i_1) + \bm x_2(i_2), &\forall i \in \N\label{eq:repair-semiring:oplus}\\
    (\bm x_1 \otimes \bm x_2)(i) &\defeq \max_{i_1, i_2 \in \N: i_1 + i_2 = i} \bm x_1(i_1) \times \bm x_2(i_2), &\forall i \in \N\label{eq:repair-semiring:otimes}
\end{align}

We will argue later that the algebraic structure  $\bm K=(K, \oplus, \otimes)$ defined above is indeed a \twomonoid (but {\em not} a semiring),
and it can be used to solve the
{\bagsetmax} problem 
using the algorithm from Eq.~\eqref{eq:example:T'}-\eqref{eq:example:Q} in polynomial time.
To that end, we also need to construct an input $\bm K$-annotated database instance, which is done as follows: (i)
every fact $f$ that is already present in the database instance $\calD$ that we want to repair
already has a multiplicity of 1 for free, hence its $\bm K$-annotation is the all-ones vector $\bm x = (1, 1, \ldots)$; (ii) in contrast, facts $f$ that are not in $\calD$ but in the repair database $\Dr$
have a multiplicity of 0 for free, but that multiplicity can be raised to one if we are willing to pay a repair cost of one or more. Hence, such facts are annotated with the vector $\bm x = (0, 1, 1, \ldots)$.

\rev{The intuition behind the $\oplus$ and $\otimes$ operators from Eq.~\eqref{eq:repair-semiring:oplus} and~\eqref{eq:repair-semiring:otimes} is as follows. Consider a disjunction of two formulas $f = f_1 \vee f_2$, and suppose that for each formula $f_j$, we already have a vector $\bm x_j$, where for any given repair budget $i_j$, the $i_j$-th entry of the vector $\bm x_j$ gives us the maximum multiplicity of $f_j$ that can be achieved within a repair budget of $i_j$. Moreover, suppose that the two formulas $f_1$ and $f_2$ are ``independent" in the sense that they don’t share any input facts, hence they have to be repaired independently. Now we ask: Given a repair budget of $i$, what is the maximum multiplicity of the disjunction $f = f_1 \vee f_2$ that we can achieve within this repair budget? To that end, we have to break down the repair budget $i$ into $i_1$ and $i_2$ such that $i = i_1 + i_2$ where we spend a budget of $i_1$ into repairing $f_1$ and $i_2$ into repairing $f_2$. The multiplicity of $f$ that we can achieve will be the sum of the two multiplicities we can achieve for $f_1$ and $f_2$, and we take the maximum over all possible ways to break $i$ down into $i_1$ and $i_2$. This gives us the formula from Eq.~\eqref{eq:repair-semiring:oplus}, which is basically the {\em convolution} of the two vectors $\bm x_1$ and $\bm x_2$ over the $(\N, \max, +)$ semiring.
A similar reasoning for conjunctions $f = f_1 \wedge f_2$ gives us Eq.~\eqref{eq:repair-semiring:otimes}, which is the convolution of $\bm x_1$ and $\bm x_2$ over the $(\N, \max, \times)$ semiring.}

Our correctness proofs and runtime analysis are not written for each of the three specific problems but rather in a generic way that utilizes the abstract algebraic properties of the \twomonoid.
This makes them applicable to all three problems and possibly more.
We describe the particular  \twomonoid that is used to solve the {\sc Shapley Value Computation} problem 
in Section~\ref{sec:algorithm}.

\section{Preliminaries}
\label{sec:prelims}

We use a capital letter $X$ to denote a variable and a small letter $x$
to denote a value for the corresponding variable $X$, where values come from  a countably infinite domain $\dom$.
We use boldface $\bm X$ to denote a set of variables and $\bm x$ to denote a tuple of values
for the corresponding set of variables $\bm X$.
In particular, a tuple $\bm x$ over a variable set $\bm X$ is a function $\bm x:\bm X\to \dom$.
Given two sets $S$ and $T$, we use $T^S$ to denote the set of all functions from $S$ to $T$.
As a special case, given a set of variables $\bm X$, we use $\dom^{\bm X}$ to denote the set of all tuples of values for the variable set $\bm X$.
Given a tuple $\bm x$ over a variable set $\bm X$ and a subset $\bm Y \subseteq \bm X$,
we use $\bm x(\bm Y)$ to denote the restriction of $\bm x$ to $\bm Y$, i.e., the tuple over $\bm Y$ obtained by
restricting the domain of $\bm x$ to $\bm Y$.
Given two tuples $\bm x$ and $\bm y$ over two disjoint variable sets $\bm X$ and $\bm Y$, we use $(\bm x, \bm y)$ to denote the tuple over $\bm X \cup \bm Y$ obtained by taking the union of the two tuples.
With some abuse of notation, we blur the line between a value $z$ of a variable $Z$
and a singleton tuple in $\dom^{\{Z\}}$.
For example, given a value $z$ of a variable $Z\notin \bm X$, we use $(\bm x, z)$ to denote the tuple over $\bm X \cup \{Z\}$ obtained by adding $z$ to $\bm x$.

An {\em atom} is an expression of the form $R(\bm X)$ where $R$ is a relation symbol and $\bm X$ is a set of variables.
A {\em schema} $\calS$ is a set of atoms.
A fact over $\calS$ is an expression of the form $R(\bm x)$ where $R(\bm X) \in \calS$
and $\bm x \in \dom^{\bm X}$.
A {\em database instance} $\calD^{\calS}$ over schema $\calS$ is a set of facts over $\calS$.
When the schema $\calS$ is clear from the context, we will simply write $\calD$ instead of $\calD^{\calS}$.
Given a schema $\calS = \{ R_1(\bm X_1),\hdots,R_m(\bm X_m) \}$, it will sometimes be convenient to write $\calD = \{R^\calD_1, \hdots, R^\calD_m\}$, where $R^\calD_i$ is the collection of $R_i$ facts in $\calD$, i.e.,~facts of the form $R_i(\bm x_i)$. We will typically suppress the superscript, writing $R_i$ instead of $R^\calD_i$ when no confusion results.
The {\em size} of a database instance $\calD$, denoted $\size{\calD}$, is the number of facts in $\calD$.

A {\em Boolean Conjunctive Query (\bcq)}, $Q$, has the following form, where $\atoms(Q)$ is the set of atoms in $Q$:
\begin{align}
    Q() \cd \bigwedge_{R(\bm X) \in \atoms(Q)} R(\bm X)
    \label{eq:sjfbcq}
\end{align}
We use $\vars(Q)$ to denote the set of all variables in $Q$, i.e., $\vars(Q) \defeq \bigcup_{R(\bm X) \in \atoms(Q)} \bm X$.
Given a variable $Y \in \vars(Q)$, we use $\atoms(Y)$ to denote the set of all atoms in $Q$ that contain $Y$, i.e. $\atoms(Y) \defeq \{R(\bm X) \in \atoms(Q) \mid Y \in \bm X\}$.
A \bcq $Q$ is said to be {\em self-join-free} if $\atoms(Q)$ have distinct relation symbols, i.e.,
for every pair of distinct atoms $R_1(\bm X), R_2(\bm X) \in \atoms(Q)$, we have $R_1 \neq R_2$.
Throughout the paper, we only consider {\em Self-Join-Free \bcqs}, abbreviated \sjfbcqs.
A database instance $\calD$ for $Q$ is a database instance over the schema $\atoms(Q)$.
Throughout the paper, we will refer to a database instance $\calD$ in the context of a specific query $Q$ to implicitly mean that $\calD$ is over the schema $\atoms(Q)$.


\section{The Bag-Set Maximization Problem}
\label{sec:bag-set-repair}

In this section, we define the \bagsetmax problem, and show that it is $\np$-complete for all non-hierarchical \sjfbcqs.
\begin{definition}[\bagsetmax]
    \label{defn:bag-set:max}
    The \bagsetmax problem is an optimization problem that is parameterized by a self-join-free Boolean conjunctive query $Q$, which is interpreted under \rev{{\em bag-set semantics}~\cite{10.5555/260873}}.
    The input to the problem is a tuple $(\calD, \Dr, \theta)$ where:
    \begin{itemize}
        \item $\calD$ is a set database instance, which is the input instance that we want to repair.
        \item $\Dr$ is another set database instance, called the {\em repair database}, which contains the facts that we are allowed to add to $\calD$ in order to repair $\calD$.
        \item $\theta$ is a natural number representing an upper bound on the number of facts
        we are allowed to add to $\calD$ in order to repair $\calD$.
    \end{itemize}
    The target is to find the maximum value of $Q(\calD')$ under \rev{bag-set semantics} over all set database instances $\calD'$ that are {\em valid repairs} of $\calD$, i.e., that satisfy:
    \begin{itemize}
        \item $\calD\subseteq\calD'\subseteq \calD \cup \Dr$ and $\card{\calD'\setminus\calD} \leq \theta$. In words, a valid repair $\calD'$
        results from $\calD$ by adding at most $\theta$ facts from $\Dr$.
        We call the quantity $\card{\calD'\setminus\calD}$ the {\em repair cost} and denote it
        $\cost(\calD, \calD')$.
    \end{itemize}
\end{definition}
The above optimization problem has a natural decision version, which we define below.
\begin{definition}[\bagsetmaxdecision]
    \label{defn:bag-set:max:decision}
    The \bagsetmaxdecision problem is the decision version of \bagsetmax.
    It is parameterized by a self-join-free Boolean conjunctive query $Q$, and has the same inputs as \bagsetmax and one additional input which is a natural number $\tau \in \mathbb{N}$.
    The target is to decide whether the maximum value of $Q(\calD')$ under \rev{bag-set semantics} over all valid repairs $\calD'$ of $\calD$ is at least $\tau$.
\end{definition}

The following proposition is straightforward.
\begin{proposition}
    For any \sjfbcq $Q$, the \bagsetmaxdecision problem is in $\np$.
    \label{prop:bag-set:max:in-np}
\end{proposition}
To prove the above proposition, we can guess a valid repair $\calD'$ of $\calD$ and verify
that $Q(\calD') \geq \tau$.
Both the size of $\calD'$ and the time needed to evaluate $Q(\calD')$ are polynomial
in the input size. (Recall that the query $Q$ is a fixed parameter to the problem, hence
its size is a constant, \rev{which means we are using data complexity.})

We now prove that for any {\em non-hierarchical} \sjfbcq $Q$, the \bagsetmaxdecision problem for $Q$ is $\np$-complete.
To prove $\np$-hardness, we rely on a reduction from the {\sc Balanced Complete Bipartite Subgraph} (\BCBS) problem, which is a well-known $\np$-complete problem, that is defined as follows:
Given a (undirected and self-loop-free) graph $G = (V,E)$ and a positive integer $k$, is there a complete bipartite subgraph of $G$ in which each of the two parts has size $k$? \BCBS is known to be $\np$-complete \cite[Problem GT24]{DBLP:books/fm/GareyJ79}.

\begin{theorem}
\label{BCBS-prop}
For any non-hierarchical \sjfbcq $Q$, the \bagsetmaxdecision problem for $Q$ is $\np$-complete.
\end{theorem}
\begin{proof}
    Membership in $\np$ follows from Proposition~\ref{prop:bag-set:max:in-np}.
    We now prove $\np$-hardness.
Let $Q$ be an arbitrary non-hierarchical \sjfbcq. 
Since $Q$ is not hierarchical, 
$Q$ must have the form:
\[
Q() \cd R(A,\bm{X}), S(A,B,\bm{Y}), T(B,\bm{Z}), P_1(\bm{D}_1), \hdots, P_n(\bm{D}_n),
\]
where $A, B$ are distinct variables and $\bm{X}$, $\bm{Y}$, $\bm{Z}$, $\bm{D}_1, \hdots, \bm{D}_n$ are sets of variables, $A \notin \bm{Z}$ and $B\notin\bm{X}$, and $n\geq 0$. We give a polynomial-time reduction from \BCBS to \bagsetmaxdecision for $Q$. Let $(G, k)$ be an input instance to \BCBS where $G=(V,E)$ is an undirected graph without self-loops and $k$ is a positive integer.
We construct an input instance $(\calD, \Dr, \theta, \tau)$ for the \bagsetmaxdecision problem for $Q$ as follows:
Let $\dom$ be the set of vertices $V$, and fix an arbitrary vertex $a \in V$.
Let $\bm W \defeq \vars(Q)$ and $\Gamma$ be the set of all tuples $\bm w \in \dom^{\bm W}$ such that $\bm w(X) = a$ for all $X \in \bm W\setminus\{ A,B \}$.
The database instance $\calD = \{R^{\calD},S^{\calD},T^{\calD},P_1^\calD,\hdots,P_n^\calD\}$ is defined as follows:
\begin{align*}
    R^\calD &= T^\calD = \emptyset,\\
    S^\calD &= \{ S(\bm w(A),\bm w(B),\bm w(\bm{Y})) \mid \bm w \in \Gamma ~\text{and}~ \{\bm w(A),\bm w(B)\} \in E \},\\
    P_i^\calD &= \{ P_i(\bm w(\bm D_i)) \mid \bm w \in \Gamma ~\text{and}~ \{\bm w(A),\bm w(B)\} \in E  \}, \quad \forall i \in [n].
\end{align*}
In contrast, the repair database $\Dr=\{R^{\Dr},S^{\Dr},T^{\Dr},{P_1}^{\Dr},\hdots,{P_n}^{\Dr}\}$ is defined as follows:
\begin{align*}
    R^{\Dr}=\{ R(\bm w(A),\bm w(\bm{X})) \mid \bm w \in \Gamma \},
    \quad
    T^{\Dr}= \{ T(\bm w(B),\bm w(\bm{Z})) \mid \bm w \in \Gamma \},
    \quad
    S^{\Dr}= P_1^{\Dr} = \hdots = P_n^{\Dr} = \emptyset.
\end{align*}
In other words, we encode the edge relation into $S^\calD$, and we permit repairs $\calD'$ of $\calD$ to be obtained by adding $R$ or $T$ facts where $A$ and $B$ are assigned to arbitrary vertices in $V$, while all other variables must be assigned to the fixed vertex $a$.
Finally, we set $\theta = 2k$ and $\tau = k^2$, thus specifying the input instance
$(\calD, \Dr, \theta, \tau)$ for \bagsetmaxdecision.

We claim that the following two statements are equivalent:
\begin{enumerate}
\item $(G, k)$ is a ``yes'' instance of \BCBS, i.e., $G$ has a complete bipartite subgraph in which each of the two parts has size $k$.
\item There is a database $\calD'$ such that $\calD \subseteq \calD' \subseteq \calD \cup \Dr$, $\cost(\calD,\calD') \leq 2k$, and $Q(\calD') \geq k^2$.
\end{enumerate}

\paragraph*{Proving $(1) \Longrightarrow (2)$:}
If $G$ has a complete balanced bipartite subgraph with parts $U_1$ and $U_2$, each of size $k$, then let $\calD' = \{R^{\calD'},S^{\calD'},T^{\calD'},P_1^{\calD'},\hdots,P_n^{\calD'}\}$ be the set database in which $S^{\calD'} = S^{\calD}$, $P_i^{\calD'} = P_i^{\calD}$ for each $i \in [n]$, and
\begin{align*}
R^{\calD'} &= \{ R(\bm w(A),\bm w(\bm{X})) \mid \bm w \in \Gamma ~\text{and}~\bm w(A) \in U_1\}, \\
T^{\calD'} &= \{ T(\bm w(B),\bm w(\bm{Z})) \mid \bm w \in \Gamma ~\text{and }~ \bm w(B) \in U_2\}.
\end{align*}
Then, $\cost(\calD,\calD')= \card{U_1} + \card{U_2} = 2k$. Let $\Gamma'$ denote the set of tuples $\bm w \in \Gamma$ such that $\bm w(A)\in U_1$ and $\bm w(B)\in U_2$.
Note that $\Gamma'$ contains exactly $k^2$ tuples, each of which represents a satisfying assignment for $Q$ in $\calD'$. Hence $Q(\calD') \geq k^2$.

\paragraph*{Proving $(2) \Longrightarrow (1)$:}
Assume that $\calD'=\{R^{\calD'},S^{\calD'},T^{\calD'},P_1^{\calD'},\hdots,P_n^{\calD'}\}$ is a set database such that $\calD \subseteq \calD' \subseteq \calD \cup \Dr$, $\cost(\calD,\calD') \leq 2k$, and $Q(\calD') \geq k^2$. By the construction of $\calD$ and $\Dr$, we must have $S^{\calD'} = S^{\calD}$ and $P_i^{\calD'} = P_i^{\calD}$ for each $i \in[n]$. By the construction of $\calD$, every satisfying assignment to $Q$ in $\calD'$ corresponds to a tuple $\bm w \in \Gamma$. Furthermore, every tuple $\bm w \in \Gamma$ is uniquely determined by its $A$ and $B$ values. Therefore, $k^2 \leq Q(\calD') \leq \card{R^{\calD'}} \times \card{T^{\calD'}}$. Since $\cost(\calD,\calD') \leq 2k$, we also have that $\card{R^{\calD'}} + \card{T^{\calD'}} \leq 2k$. By a straightforward calculus argument, this implies that $\card{R^{\calD'}} = \card{T^{\calD'}} = k$. Now let
\begin{align*}
U_1 &= \{ \bm w(A) \in V \mid \bm w \in \Gamma ~\text{and}~ R(\bm w(A),\bm w(\bm{X})) \in R^{\calD'}  \}  \\
U_2 &= \{ \bm w(B) \in V \mid \bm w \in \Gamma ~\text{and}~ T(\bm w(B),\bm w(\bm{Z})) \in T^{\calD'}  \}.
\end{align*}
For $Q(\calD')$ to be at least $k^2$, for every $(v_1, v_2) \in U_1\times U_2$, we must have $S(v_1,v_2,\bm w(\bm{Y})) \in S^{\calD'}=S^{\calD}$, which implies that $\{v_1,v_2\} \in E$. Since $G$ does not contain self-loops, it follows that $v_1\neq v_2$ for all $(v_1, v_2) \in U_1\times U_2$. Hence, $U_1$ and $U_2$ form a complete bipartite subgraph of size $k\times k$ in $G$.
\end{proof}


We conclude this section with some immediate corollaries of Theorem~\ref{BCBS-prop} regarding the parameterized complexity and approximation hardness of \bagsetmax.

One natural parameterization of \bagsetmaxdecision treats both the repair cost $\theta$ and the target $\tau$ as parameters. An alternative parameterization is the \emph{standard parameterization} \cite{flum2006parameterized}, where only the target $\tau$ is a parameter. In either case, the polynomial-time reduction in Theorem \ref{BCBS-prop} can be extended to an $\fpt$-reduction from the parameterized \BCBS problem -- parameterized by the size $k$ of each part of the desired balanced complete bipartite subgraph -- by observing that $\theta$ and $\tau$ in the resulting \bagsetmaxdecision instance are computable functions of $k$. Hence by the $\wone$-hardness of \BCBS \cite{DBLP:journals/jacm/Lin18}, we obtain the following.

\begin{corollary}
For \rev{every} non-hierarchical \sjfbcq $Q$, the parameterized \bagsetmaxdecision problem for $Q$ (with parameters $\theta$ and $\tau$ or with just parameter $\tau$) is $\wone$-hard.
\end{corollary}

Additionally, the \bagsetmax problem can be interpreted as an approximation problem by taking the set of feasible solutions for an instance $(\calD, \Dr, \theta)$ of the problem to be any repair $\calD'$ of repair cost at most $\theta$. The \emph{value} associated to this instance and a feasible repair $\calD'$ is $Q(\calD')$. While hardness-of-approximation results for \BCBS are known \cite{feige2004hardness}, it is not obvious how to extend the reduction of Theorem \ref{BCBS-prop} to an appropriate approximation-preserving reduction. However, it is known that, if an $\np$-optimization problem has a {\em fully polynomial-time approximation scheme} (\fptas), then its standard parameterization is in $\fpt$ \cite{cai1997fixed}. Hence, by the $\wone$-hardness of its standard parameterization, we obtain the following.

\begin{corollary}
For \rev{every} non-hierarchical \sjfbcq $Q$, the \bagsetmax problem for $Q$ does not have an \fptas unless $\fpt = \wone$.
\end{corollary}

\section{The Unifying Algorithm for Hierarchical Queries}
\label{sec:algorithm}

In order to describe our general-purpose polynomial-time algorithm for hierarchical \sjfbcqs, we need to develop some technical tools.

\subsection{An Elimination Procedure for Hierarchical Queries}
\rev{As an alternative to the standard definition from the introduction,} hierarchical queries can be defined in terms of a specific ``elimination procedure'',
where a query $Q$ is hierarchical if and only if this elimination procedure eliminates all its variables, thus reducing it to a query of the form $Q() \cd R()$.
This elimination procedure will be at the heart of our polynomial-time algorithm
where the algorithm's steps mirror the elimination steps.
The elimination procedure is also similar in spirit to the {\em GYO elimination procedure} for acyclic queries~\cite{AbiteboulHV95},
\rev{and also runs in polynomial time in the size of the query, just like the GYO procedure does.}

\begin{proposition}[Elimination Procedure for Hierarchical Queries]
    \label{prop:hq-elimination}
    A \sjfbcq $Q$ is hierarchical if and only if the following elimination procedure transforms $Q$ into a query with a single nullary atom, i.e., $Q() \cd R()$. Repeatedly, apply either one of the following two rules, until neither rule applies:
    \begin{itemize}
        \item {\em (Rule 1)} Eliminate a ``private'' variable $Y$: Find a variable $Y \in \vars(Q)$ that occurs in only one atom $R(\bm X)$ in $Q$, and replace $R(\bm X)$ with $R'(\bm X \setminus \{Y\})$ where $R'$ is a new relation symbol.
        \item {\em (Rule 2)} Eliminate a ``duplicate atom'': Find two distinct atoms $R_1(\bm X)$ and $R_2(\bm X)$
        that share the same set of variables $\bm X$, and replace both atoms with a single atom $R'(\bm X)$ where $R'$ is a new relation symbol.\footnote{Note that we are {\em not} adding two copies of the new atom $R'(\bm X)$ but only a single copy.
        Hence, the resulting query is still self-join-free.}
    \end{itemize}
    Moreover, the above two rules maintain the hierarchical property.
    In particular, after each application of Rule 1 or Rule 2, the resulting query is hierarchical
    if and only if the original query was hierarchical.
\end{proposition}

Note that the above elimination procedure is very similar in nature to the {GYO elimination procedure}~\cite{AbiteboulHV95},
which is typically used to detect if a query is {acyclic}.
The only difference is that in GYO, Rule 2 is more relaxed: We find two atoms $R_1(\bm X)$ and $R_2(\bm Y)$ where $\bm X \subseteq \bm Y$, and we replace them with a single atom $R'(\bm Y)$.
As a result, whenever the above hierarchical elimination procedure succeeds on a query,
the GYO elimination procedure is guaranteed to succeed on the same query (but the opposite is not necessarily true).
This in turn is consistent with the well-known fact that hierarchical queries are a (strict) subset of acyclic queries.
Moreover, just like GYO, there could be multiple ways to apply the elimination procedure to a given query,
but they are all guaranteed to lead to the same conclusion, regarding the hierarchical property of the query.
\rev{Before proving Proposition~\ref{prop:hq-elimination}, we give a few examples.}


\begin{example}
    \label{eq:example:elim:1}
    Consider the query $Q$ from Eq.~\eqref{eq:example:query}.
    Here is one way to apply the elimination procedure. It shows that the query is hierarchical. Any other way to apply the procedure will lead to the same conclusion.
    \begin{align*}
        Q() &\cd R(A, B) \wedge S(A, C) \wedge {\color{blue}T(A, C, D)} & \text{(Rule 1)}\\
        Q() &\cd R(A, B) \wedge {\color{blue}S(A, C) \wedge T'(A, C)} & \text{(Rule 2)}\\
        Q() &\cd R(A, B) \wedge {\color{blue}S'(A, C)} & \text{(Rule 1)}\\
        Q() &\cd {\color{blue}R(A, B)} \wedge S''(A) & \text{(Rule 1)}\\
        Q() &\cd {\color{blue}R'(A) \wedge S''(A)} & \text{(Rule 2)}\\
        Q() &\cd {\color{blue}R''(A)} & \text{(Rule 1)}\\
        Q() &\cd {\color{blue}R'''()} & \text{(Done!)}
    \end{align*}
    Moreover, each one of the above intermediate queries is also hierarchical.
\end{example}

\begin{example}
    Consider the following query:
    \begin{align*}
        Q() \cd R(A, B) \wedge S(B, C) \wedge T(C, D)
    \end{align*}
    The elimination procedure can proceed as follows:
    \begin{align*}
        Q() &\cd R(A, B) \wedge S(B, C) \wedge {\color{blue}T(C, D)}& \text{(Rule 1)}\\
        Q() &\cd {\color{blue}R(A, B)} \wedge S(B, C) \wedge T'(C)& \text{(Rule 1)}\\
        Q() &\cd R'(B) \wedge S(B, C) \wedge T'(C)& \text{(Stuck!)}\\
    \end{align*}
    But now we are left with a query where neither rule applies. This shows that the original query
    and all of the above intermediate queries are {\em not} hierarchical.
\end{example}

The following example shows that even if the input hierarchical query was disconnected, the elimination procedure still reduces it to a {\em single} nullary atom.
\begin{example}
    Consider the query:
    \begin{align*}
        Q() \cd R(A) \wedge S(B)
    \end{align*}
    The elimination procedure proceeds as follows:
    \begin{align*}
        Q() &\cd R(A) \wedge {\color{blue}S(B)}& \text{(Rule 1)}\\
        Q() &\cd {\color{blue}R(A)} \wedge S'()& \text{(Rule 1)}\\
        Q() &\cd {\color{blue}R'() \wedge S'()}& \text{(Rule 2)}\\
        Q() &\cd R''() & \text{(Done!)}\\
    \end{align*}
    The query is hierarchical.
\end{example}

\rev{In order to prove Proposition~\ref{prop:hq-elimination}, we need some additional definitions and results.}
Given a \sjfbcq $Q$, two atoms $R_1(\bm X_1), R_2(\bm X_2)\in\atoms(Q)$ are called {\em connected} if:
\begin{itemize}
    \item they share at least one variable, i.e., $\bm X_1 \cap \bm X_2 \neq \emptyset$, or
    \item they are both connected to a third atom $R_3(\bm X_3)$.
\end{itemize}
A \sjfbcq $Q$ is called {\em connected} if every pair of atoms in $Q$ are connected.
Given a \sjfbcq $Q$, the {\em connected components} of $Q$ are the (unique) connected \sjfbcqs $Q_1, Q_2, \ldots, Q_m$
that satisfy $Q = Q_1 \wedge Q_2 \wedge \ldots \wedge Q_m$ and $\vars(Q_i) \cap \vars(Q_j) = \emptyset$ for all $i \neq j$.
The following proposition is well-known~\cite{10.1145/2395116.2395119}.

\begin{proposition}[\cite{10.1145/2395116.2395119}]
    \label{prop:hierarchical:trie}
    A connected \sjfbcq $Q$ is hierarchical if and only if there exists a rooted tree $T$ whose set of nodes is exactly $\vars(Q)$ that satisfies the following: For every atom $R(\bm X)\in\atoms(Q)$, there exists a node $Y\in \vars(Q)$ where the set of variables along the path from $Y$ to the root of $T$ (including $Y$ and the root) is exactly $\bm X$.
\end{proposition}

\begin{proof}[Proof of Proposition~\ref{prop:hq-elimination}]
    To prove Proposition~\ref{prop:hq-elimination}, we will prove the following claims:
    \begin{claim}
        \label{clm:hq-elimination:1}
        Given a hierarchical \sjfbcq $Q$ that does not have the form $Q() \cd R()$, either Rule 1 or Rule 2 must be applicable to $Q$, and the resulting query must be hierarchical.
    \end{claim}
    \begin{claim}
        \label{clm:hq-elimination:2}
        Given a non-hierarchical \sjfbcq $Q$, if either Rule 1 or Rule 2 is applicable to $Q$, then the resulting query must be non-hierarchical.
    \end{claim}
    \begin{proof}[Proof of Claim~\ref{clm:hq-elimination:1}]
        Claim~\ref{clm:hq-elimination:1} follows from Proposition~\ref{prop:hierarchical:trie}.
        In particular, consider a hierarchical \sjfbcq $Q$ that does not have the form $Q() \cd R()$.
        If $\vars(Q) = \emptyset$, then all connected components of $Q$ must have the form $Q_i() \cd R_i()$,
        in which case we can apply Rule 2 on them.
        On the other hand, if $\vars(Q) \neq \emptyset$, let $Q_i$ be a connected component of $Q$ that
        has at least one variable.
        Let $T_i$ be the tree that satisfies Proposition~\ref{prop:hierarchical:trie}, let $Y \in \vars(Q)$ be any leaf in the tree, and let $\bm X$ be the set of variables from $Y$ to the root of $T_i$.
        We recognize two cases:
        \begin{itemize}
            \item If $Q_i$ contains more than one atom whose variables are $\bm X$, then we can apply Rule 2 on them.
            \item Otherwise, let $R(\bm X)$ be the unique atom in $Q_i$ with a variable set $\bm X$.
            In this case, $R(\bm X)$ must be the only atom in $Q_i$ that contains $Y$, in which case
            we can apply Rule 1.
        \end{itemize}
        In both cases above, the resulting component $Q_i$ still has a tree $T_i'$ satisfying Proposition~\ref{prop:hierarchical:trie}, hence the resulting $Q$ is still hierarchical. This proves Claim~\ref{clm:hq-elimination:1}.
    \end{proof}

    \begin{proof}[Proof of Claim~\ref{clm:hq-elimination:2}]
        In order to prove Claim~\ref{clm:hq-elimination:2}, consider a non-hierarchical \sjfbcq $Q$.
        \rev{By definition of hierarchical queries,} $Q$ must contain two distinct variables $A, B \in \vars(Q)$ and three distinct
        atoms $R(\bm X), S(\bm Y), T(\bm Z)\in\atoms(Q)$ such that:
        \begin{itemize}
            \item $A$ is in $\bm X$ and $\bm Y$ but not in $\bm Z$, and
            \item $B$ is in $\bm Y$ and $\bm Z$ but not in $\bm X$.
        \end{itemize}
        Suppose that Rule 1 is applicable to $Q$.
        Note that the private variable that we eliminate in Rule 1 cannot be either $A$ or $B$ because
        each one of them is shared by at least two atoms.
        Hence after applying Rule 1, we can still use the same two variables $A$ and $B$ along with the atoms $R(\bm X), S(\bm Y), T(\bm Z)$ (one of which might be the one containing the private variable) to show that the resulting query is still non-hierarchical.

        Now suppose that Rule 2 is applicable to $Q$.
        Note that the two duplicate atoms from Rule 2 that share the same set of variables cannot both belong
        to the set of atoms $\{R(\bm X), S(\bm Y), T(\bm Z)\}$ because none of these three atoms share the same set of variables.
        Hence, even after applying Rule 2, we will still be left with three distinct atoms with variable sets $\bm X, \bm Y, \bm Z$. We can still use these atoms along with the two variables $A$ and $B$ to show that the resulting query is still non-hierarchical. This proves Claim~\ref{clm:hq-elimination:2}.
    \end{proof}
    We now show how to use Claims~\ref{clm:hq-elimination:1} and \ref{clm:hq-elimination:2} to prove Proposition~\ref{prop:hq-elimination}.
    Note that applying Rule 1 reduces the number of variables in the query by one, while applying Rule 2 reduces the number of atoms by one.
    Therefore, the two rules can only be applied a finite number of times.
    Claim~\ref{clm:hq-elimination:1} guarantees that if the query is hierarchical,
    then the two rules can be repeatedly applied until the query is reduced to the form $Q() \cd R()$.
    On the other hand, Claim~\ref{clm:hq-elimination:2} guarantees that if the query is non-hierarchical, then the two rules are going to preserve it as non-hierarchical,
    hence it will never be reduced to the form $Q() \cd R()$, since the latter is a hierarchical query.
    This proves Proposition~\ref{prop:hq-elimination}.
\end{proof}

\subsection{The Underlying Algebraic Structure: \twomonoid}
Our algorithm operates on a general algebraic structure called {\em \twomonoid}, which we define below.
A specific instantiation of this structure is chosen depending on the specific application domain
that the algorithm is used for.

A {\em \twomonoid} is a structure $\bm K = (K, \oplus, \otimes)$ that satisfies the commutative semiring properties, with two exceptions:
(a) it doesn't necessarily satisfy the distributive property, and (b) the annihilation-by-zero property is replaced by the weaker property $\bm 0 \otimes \bm 0 = \bm 0$.
\begin{definition}[\twomonoid]
    \label{defn:twomonoid}
     A triple $\bm K = (K, \oplus, \otimes)$ is a {\em \twomonoid} iff it satisfies the following:
    \begin{itemize}
        \item $(K, \oplus)$ is a commutative monoid\footnote{\rev{A {\em monoid} $(K, \oplus)$ is a pair where $K$ is a set and $\oplus$ is an associative binary operator over the elements of $K$ that has an {\em identity element} $\bm 0$, i.e., an element $\bm 0\in K$ where $a \oplus \bm 0 = \bm 0\oplus a=a$, for all $a \in K$. A monoid $(K,\oplus)$ is {\em commutative} if $\oplus$ is commutative.}} with identity element $\bm 0$.
        \item $(K, \otimes)$ is a commutative monoid with identity element $\bm 1$.
        \item \rev{The identity element, $\bm 0$, of $(K, \oplus)$ satisfies} $\bm 0 \otimes \bm 0 = \bm 0$.
    \end{itemize}
\end{definition}
We illustrate the \twomonoid structure for the three problems we consider in Sections~\ref{subsec:alg:probabilistic},~\ref{subsec:alg:bagsetmax}, and~\ref{subsec:alg:shapley}.
\footnote{\rev{Note that if $\bm 0 = \bm 1$, then the \twomonoid becomes trivial.
Nevertheless, our claims and proofs still hold, which is why we don't explicitly exclude this case.}}

\subsection{The Algorithm}
We now present our general algorithm for hierarchical queries.
Depending on the application, the algorithm can be instantiated to compute different values
representing the solutions to different problems.
The algorithm is instantiated by providing a specific instantiation of the \twomonoid $\bm K = (K, \oplus, \otimes)$ as a parameter.
The algorithm is depicted in Algorithm~\ref{alg:hierarchical}.
It takes as input a hierarchical \sjfbcq $Q$ and a $\bm K$-annotated database instance $\calD$.
The algorithm mirrors the elimination procedure from Prop.~\ref{prop:hq-elimination}.
As an example, note that the specific algorithm from Eq.~\eqref{eq:example:T'}-\eqref{eq:example:Q} mirrors the elimination steps from Example~\ref{eq:example:elim:1}.
The algorithm repeatedly applies either Rule~1 or Rule~2 until the query is reduced to the form $Q() \cd R()$.
We use the $\oplus$-operator in Rule 1 and the $\otimes$-operator in Rule 2 to compute new $\bm K$-annotated relations that replace old ones.
When the elimination procedure stops, the remaining query is guaranteed to have the form $Q() \cd R()$, assuming it was hierarchical to begin with. We return the $\bm K$-annotation of the nullary tuple $()$ in $R$ as the output of the algorithm.

\begin{algorithm}[th!]
    \caption{General-Purpose Algorithm for Hierarchical Queries}
    \label{alg:hierarchical}
    \begin{algorithmic}[1]
        \Statex{\textbf{Parameters:}}
        \begin{itemize}
            \item A \twomonoid $\bm K = (K, \oplus, \otimes)$.
        \end{itemize}
        \Statex{\textbf{Input:}}
        \begin{itemize}
            \item A hierarchical \sjfbcq $Q$ (Eq.~\eqref{eq:sjfbcq}).
            \item A $\bm K$-annotated database instance $\calD$.
            \begin{itemize}
                \item For each $R(\bm X) \in \atoms(Q)$, the relation $R$ is $\bm K$-annotated in $\calD$.
            \end{itemize}
        \end{itemize}
        \Statex{\textbf{Output:}}
        \begin{itemize}
            \item A value $q \in K$. The semantics for $q$ depends on the instantiation of the algorithm.
        \end{itemize}
        \Statex \rule{\linewidth}{0.4pt}
        \While{$Q$ is not of the form $Q() \cd R()$}
            \Comment Apply Rule 1 or 2 from Prop.~\ref{prop:hq-elimination}
            \If{$\exists Y \in \vars(Q)$ where $Y$ occurs in only one atom $R(\bm X)$}\Comment{Rule 1} \label{alg:meta:rule1}
                \State Let $\bm X'\defeq \bm X \setminus \{Y\}$, hence $R(\bm X) = R(\bm X', Y)$
                \State Create a new atom $R'(\bm X')$ defined as $R'(\bm x') \defeq \bigoplus_{y \in \dom} R(\bm x', y)$, for all $\bm x' \in \dom^{\bm X'}$
                \State Replace $R(\bm X)$ with $R'(\bm X')$
                \label{alg:meta:rule1:replace}
            \ElsIf{$\exists$ two atoms $R_1(\bm X)$ and $R_2(\bm X)$ with the same set of variables $\bm X$}\Comment{Rule 2} \label{alg:meta:rule2}
                \State Create a new atom $R'(\bm X)$ defined as $R'(\bm x) \defeq R_1(\bm x) \otimes R_2(\bm x)$, for all $\bm x \in \dom^{\bm X}$
                \State Replace both $R_1(\bm X)$ and $R_2(\bm X)$ with a single atom $R'(\bm X)$
                \label{alg:meta:rule2:replace}
            \Else
                \State{Report $Q$ is not hierarchical!} \Comment{By Prop.~\ref{prop:hq-elimination}}
            \EndIf
        \EndWhile
        \Return $R()$, i.e., the $\bm K$-annotation of the tuple $()$ in $R$ \Comment{$Q$ now has the form $Q() \cd R()$}
    \end{algorithmic}
\end{algorithm}

\subsection{First Instantiation: {\sc Probabilistic Query Evaluation}}
\label{subsec:alg:probabilistic}
We now instantiate Algorithm~\ref{alg:hierarchical} to compute the probability of a query $Q$ evaluating to true over a tuple-independent probabilistic database $\calD$.
The \twomonoid we use is defined as follows:
\begin{definition}[\twomonoid for {\sc Probabilistic Query Evaluation}]
    \label{defn:prob-semiring}
    We define a specific \twomonoid $\bm K = (K, \oplus, \otimes)$ where
    the domain $K$ is the set of all probabilities, i.e. $K \defeq [0, 1]$, and the $\oplus$ and $\otimes$ are defined using Eq.~\eqref{eq:prob-semiring:oplus} and~\eqref{eq:prob-semiring:otimes} respectively, for every $p_1, p_2 \in [0, 1]$.
    The identities for $\oplus$ and $\otimes$ are $\bm 0 \defeq 0$ and $\bm 1 \defeq 1$ respectively.
\end{definition}

We prove the following theorem in the next section.
\begin{theorem}
    Let $\bm K = (K, \oplus, \otimes)$ be the \twomonoid from Definition~\ref{defn:prob-semiring}.
    Let $Q$ be a hierarchical \sjfbcq, and $\calD$ be a tuple-independent probabilistic database, viewed as 
    a $\bm K$-annotated database instance $\calD$, where each fact is annotated with its probability.
    Then using $Q$ and $\calD$ as inputs, Algorithm~\ref{alg:hierarchical} runs in time $O(\size{\calD})$
    and outputs the probability of $Q$ evaluating to true over all possible worlds.
    \label{thm:correct:probabistic}
\end{theorem}

\subsection{Second Instantiation: \bagsetmax}
\label{subsec:alg:bagsetmax}

In order to use Algorithm~\ref{alg:hierarchical} to solve \bagsetmax, we define a suitable \twomonoid.

\begin{definition}[\twomonoid for \bagsetmax]
    \label{defn:repair-semiring}
    We define a specific \twomonoid $\bm K = (K, \oplus, \otimes)$ as follows. The domain $K$ is the set of all vectors of natural numbers
    $\N^{\N}$ that are {\em monotonic}\footnote{Recall that we use the notation $T^S$ to denote the set of all functions from $S$ to $T$.
    Hence, $\N^\N$ is the set of all functions from $\N$ to $\N$, or equivalently vectors of natural numbers.}, i.e.,
    $K \defeq \left\{\bm x \in \N^{\N} \mid \forall i < j, \bm x(i) \leq \bm x(j)\right\}$.
    For every pair of vectors $\bm x_1, \bm x_2 \in K$,
    the $\oplus$ and $\otimes$ operators are defined 
    using Eq.~\eqref{eq:repair-semiring:oplus} and~\eqref{eq:repair-semiring:otimes} respectively.
    The identities $\bm 0$ and $\bm 1$ for $\oplus$ and $\otimes$ are the all-zeros and all-ones vectors respectively,
    i.e., $\bm 0(i) \defeq 0$ and $\bm 1(i) \defeq 1$ for all $i \in \N$.
\end{definition}

It is straightforward to verify that the structure $\bm K = (K, \oplus, \otimes)$ from Definition~\ref{defn:repair-semiring} satisfies all \twomonoid properties from Definition~\ref{defn:twomonoid}.

Given an input instance $(\calD, \Dr, \theta)$ to the \bagsetmax problem from Definition~\ref{defn:bag-set:max},
the following definition helps us construct a $\bm K$-annotated database instance $\psi(\calD, \Dr)$ that is used as input to Algorithm~\ref{alg:hierarchical}.
To that end, we define a vector $\bm \star \in K$ as follows. For any $i \in\N$:
\begin{align*}
    \bm \star(i) \defeq
    \begin{cases}
        0 & \text{if $i = 0$}\\
        1 & \text{otherwise}
    \end{cases}
\end{align*}
\begin{definition}[Input $\bm K$-annotated database instance for \bagsetmax]
    Let $\bm K = (K, \oplus, \otimes)$ be the \twomonoid from Definition~\ref{defn:repair-semiring}, $Q$ be a $\sjfbcq$, and $(\calD, \Dr, \theta)$ be an input instance to the \bagsetmax problem from Definition~\ref{defn:bag-set:max}.
    We define a function $\psi$ that maps a fact $R(\bm x)$ where $R(\bm X) \in \atoms(Q)$ and $\bm x \in \dom^{\bm X}$
    to a monotonic vector in $K$ as follows:
    \begin{align}
        \psi(R(\bm x)) \defeq
        \begin{cases}
            \bm 1 & \text{if $R(\bm x)\in \calD$}\\
            \bm \star & \text{if $R(\bm x)\notin\calD$ and $R(\bm x)\in\Dr$}\\
            \bm 0 & \text{otherwise}
        \end{cases}
    \end{align}
    Moreover, we define $\psi(\calD, \Dr)$ as a $\bm K$-annotated database instance where
    for every $R(\bm X) \in \atoms(Q)$ and $\bm x \in \dom^{\bm X}$, the fact $R(\bm x)$ is annotated with $\psi(R(\bm x))$.
    \label{defn:repair-annotated-database}
\end{definition}

We prove the following theorem in the next section.
\begin{theorem}
    Let $\bm K = (K, \oplus, \otimes)$ be the \twomonoid from Definition~\ref{defn:repair-semiring}, $Q$ be a hierarchical \sjfbcq, $(\calD, \Dr, \theta)$ be an input instance to \bagsetmax from Definition~\ref{defn:bag-set:max},
    and $\psi(\calD, \Dr)$ be the $\bm K$-annotated database instance from Definition~\ref{defn:repair-annotated-database}.
    Then using $Q$ and $\psi(\calD, \Dr)$ as inputs, Algorithm~\ref{alg:hierarchical}
    runs in time $O((\size{\calD} +\size{\Dr})\cdot \size{\Dr}^2)$
    \rev{and space $O((\size{\calD} +\size{\Dr})\cdot \size{\Dr})$}
    and
    outputs a vector $q \in K$ that satisfies the following:
    \begin{itemize}
        \item For every $i \in \N$, the $i$-the entry, $q(i)$, is the maximum multiplicity of $Q$ that can be achieved within a repair cost at most $i$.
    \end{itemize}
    \label{thm:correct:repair}
\end{theorem}
In particular, the output to the \bagsetmax problem is $q(\theta)$, i.e.,~the $\theta$-th entry of $q$.

\subsection{Third Instantiation: {\sc Shapley Value Computation}}
\label{subsec:alg:shapley}
We start with a brief overview of the concept of Shapley values of facts in a database,
and then explain how to instantiate Algorithm~\ref{alg:hierarchical} to compute these values.

\begin{definition}[Shapley value of a fact with respect to a query and a database]
    Let $Q$ be a Boolean Conjunctive Query and $\calD$ be a database instance.
    The instance $\calD$ is divided into two parts:
    \begin{itemize}
        \item {\em Exogenous part}, denoted $\Dexo$. These are facts that are fixed.
        \item {\em Endogenous part}, denoted $\Dendo$. These are facts that can be added or deleted from $\calD$.
    \end{itemize}
    Suppose that we pick a permutation of the facts in $\Dendo$ uniformly at random, and we begin to
    insert those facts one by one into $\Dexo$.
    The {\em Shapley value of a given fact $f$ with respect to $Q$ and $\calD$} 
    is the probability of adding the fact $f$ flipping the answer of $Q$ from $\false$ to $\true$.
    Formally, let $\pi(\Dendo)$ denote the set of all permutations of facts in $\Dendo$.
    Given a permutation $\sigma \in \pi(\Dendo)$ and a fact $f$ in $\Dendo$, let $\sigma_{<f}$ denote the
    set of facts that appear before $f$ in $\sigma$.
    Given a fact $f$ in $\Dendo$, the {\em Shapley value of $f$ with respect to $Q$ and $\calD$} is defined as:
    \begin{align}
        \shapley_{Q, \Dexo, \Dendo}(f) \defeq \frac{1}{|\Dendo|!}\sum_{\sigma \in \pi(\Dendo)}
            \left(Q(\Dexo \cup \sigma_{<f} \cup \{f\}) - Q(\Dexo \cup \sigma_{<f})\right)
    \end{align}
\end{definition}

Recent work~\cite{DBLP:journals/lmcs/LivshitsBKS21} has shown a reduction from computing $\shapley_{Q, \Dexo, \Dendo}(f)$
to computing the quantity $\sat_{Q, \Dexo, \Dendo}(k)$ that is defined below.

\begin{definition}[$\sat_{Q, \Dexo, \Dendo}(k)$]
    \label{defn:sharpsat}
    Let $Q$ be a Boolean Conjunctive Query, and $\calD$ be a database instance that is divided into
    an exogenous part $\Dexo$ and an endogenous part $\Dendo$. Given a number $k \in \N$,
    we define $\sat_{Q, \Dexo, \Dendo}(k)$ as the number of possible subsets $\calD' \subseteq \Dendo$
     of size $|\calD'| = k$ such that $Q(\Dexo \cup \calD')$ is true.
\end{definition}

\rev{We summarize the reduction by~\cite{DBLP:journals/lmcs/LivshitsBKS21} from $\shapley_{Q, \Dexo, \Dendo}(f)$
to $\sat_{Q, \Dexo, \Dendo}(k)$:}
\begin{align*}
    \shapley_{Q, \Dexo, \Dendo}(f)
    =&\sum_{\calD' \subseteq \Dendo \setminus \{f\}}
    \frac{|\calD'|! \cdot (|\Dendo| - |\calD'|-1)!}{|\Dendo|!}\cdot
    \left(Q(\Dexo \cup \calD' \cup \{f\}) - Q(\Dexo \cup \calD')\right)\\
    =&\sum_{\calD' \subseteq \Dendo \setminus \{f\}}\frac{|\calD'|! \cdot (|\Dendo| - |\calD'|-1)!}{|\Dendo|!}\cdot Q(\Dexo \cup \calD' \cup \{f\})\\
    -&\sum_{\calD' \subseteq \Dendo \setminus \{f\}}\frac{|\calD'|! \cdot (|\Dendo| - |\calD'|-1)!}{|\Dendo|!}\cdot Q(\Dexo \cup \calD')\\
    =&\sum_{k=0}^{|\Dendo|-1}\frac{k! \cdot (|\Dendo| - k - 1)!}{|\Dendo|!}\cdot
    \sat_{Q, \Dexo \cup\{f\}, \Dendo\setminus\{f\}}(k)\\
    -&\sum_{k=0}^{|\Dendo|-1}\frac{k! \cdot (|\Dendo| - k - 1)!}{|\Dendo|!}\cdot
    \sat_{Q, \Dexo, \Dendo\setminus \{f\}}(k)\\
    =&\sum_{k=0}^{|\Dendo|-1}\frac{k! \cdot (|\Dendo| - k - 1)!}{|\Dendo|!}\cdot
    \left(\sat_{Q, \Dexo\cup\{f\}, \Dendo\setminus \{f\}}(k)-\sat_{Q, \Dexo, \Dendo\setminus \{f\}}(k)\right)
\end{align*}
We now explain how to instantiate Algorithm~\ref{alg:hierarchical} to compute $\sat_{Q, \Dexo, \Dendo}(k)$.
Expectedly, the instantiation of our algorithm in this case is similar in spirit to the algorithm from~\cite{DBLP:journals/lmcs/LivshitsBKS21}:
Our Algorithm~\ref{alg:hierarchical} uses the \twomonoid abstraction
to {\em generalize}
the algorithm from~\cite{DBLP:journals/lmcs/LivshitsBKS21} to other application domains.
The specific \twomonoid we use in this case is defined as follows.

\begin{definition}[\twomonoid for computing $\sat_{Q, \Dexo, \Dendo}(k)$]
    We define a specific \twomonoid $\bm K = (K, \oplus, \otimes)$ as follows.
    Let $\B \defeq \{\true, \false\}$.
    The domain $K$ is $\N^{\N\times \B}$, i.e., the set of all vectors of natural numbers
    that are indexed by pairs of natural numbers and Boolean values.
    The $\oplus$ and $\otimes$ operations are defined as follows for any pair of vectors $\bm x_1, \bm x_2 \in K$:
    \begin{align}
        (\bm x_1 {\color{red}\oplus} \bm x_2)(i, b) &\defeq
        \sum_{i_1, i_2\in \N: i_1 + i_2 = i}\quad\sum_{b_1, b_2 \in \B: b_1 {\color{red}\vee} b_2 = b}
        \bm x_1(i_1, b_1) \times \bm x_2(i_2, b_2), &\forall i \in \N, b \in \B\label{eq:shapley:oplus}\\
        (\bm x_1 {\color{red}\otimes} \bm x_2)(i, b) &\defeq
        \sum_{i_1, i_2\in \N: i_1 + i_2 = i}\quad\sum_{b_1, b_2 \in \B: b_1 {\color{red}\wedge} b_2 = b}
        \bm x_1(i_1, b_1) \times \bm x_2(i_2, b_2), &\forall i \in \N, b \in \B\label{eq:shapley:otimes}
    \end{align}
    The identities $\bm 0$ and $\bm 1$ for $\oplus$ and $\otimes$ are as follows:
    \begin{align*}
        \bm 0(i, b) \defeq
        \begin{cases}
            1 &\text{if } i = 0 \text{ and } b = \false\\
            0 &\text{otherwise}
        \end{cases}
        \quad\quad\quad\quad\quad
        \bm 1(i, b) \defeq
        \begin{cases}
            1 &\text{if } i = 0 \text{ and } b = \true\\
            0 &\text{otherwise}
        \end{cases}
    \end{align*}
    \label{defn:shapley-semiring}
\end{definition}

Note that the above \twomonoid does {\em not} satisfy the annihilation-by-zero property:
$a \otimes \bm 0 = \bm 0$ for all $a \in K$.
It does however satisfy the weaker property $\bm 0 \otimes \bm 0 = \bm 0$ that is required
by Definition~\ref{defn:twomonoid}.

Given $Q$, $\Dexo$ and $\Dendo$, the following definition constructs a $\bm K$-annotated database instance $\psi(\Dexo, \Dendo)$ that is used as input to Algorithm~\ref{alg:hierarchical}.
In addition to the identity vectors $\bm 0$ and $\bm 1$, we also need to use the following vector $\star \in K$:
\begin{align*}
    \bm \star(i, b) &\defeq
    \begin{cases}
        1 &\text{if ($i = 0$ and $b = \false$) or ($i = 1$ and $b = \true$)}\\
        0 &\text{otherwise}
    \end{cases}
\end{align*}
\begin{definition}[Input $\bm K$-annotated database instance for computing $\sat_{Q, \Dexo, \Dendo}(k)$]
    Let $\bm K = (K, \oplus, \otimes)$ be the \twomonoid from Definition~\ref{defn:shapley-semiring},
    $Q$ be a \sjfbcq and $\Dexo$ and $\Dendo$ be two database instances for $Q$.
    We define a function $\psi$ that maps a fact $R(\bm x)$ where $R(\bm X)\in\atoms(Q)$
    and $\bm x \in \dom^{\bm X}$ to a vector in $K$ as follows:
    \begin{align}
        \psi(R(\bm x)) \defeq
        \begin{cases}
            \bm 1 & \text{if $R(\bm x) \in \Dexo$}\\
            \bm \star & \text{if $R(\bm x) \notin \Dexo$ and $R(\bm x) \in \Dendo$}\\
            \bm 0 & \text{otherwise}
        \end{cases}
    \end{align}
    Moreover, we define $\psi(\Dexo, \Dendo)$ as a $\bm K$-annotated database instance where for every $R(\bm X)\in\atoms(Q)$
    and $\bm x \in \dom^{\bm X}$, the fact $R(\bm x)$ is annotated with $\psi(R(\bm x))$.
    \label{defn:shapley-annotated-database}
\end{definition}

\rev{The following theorem will be  proved in the next section.}
\begin{theorem}
    Let $\bm K = (K, \oplus, \otimes)$ be the \twomonoid from Definition~\ref{defn:shapley-semiring},
    $Q$ be a hierarchical \sjfbcq,
    $\Dexo$ and $\Dendo$ be two database instances for $Q$,
    and $\psi(\Dexo, \Dendo)$ be the $\bm K$-annotated database instance from Definition~\ref{defn:shapley-annotated-database}.
    Then using $Q$ and $\psi(\Dexo, \Dendo)$ as inputs,
    Algorithm~\ref{alg:hierarchical} runs in time in $O((\size{\Dexo} + \size{\Dendo})\cdot \size{\Dendo}^2)$
    \rev{and space $O((\size{\Dexo} + \size{\Dendo})\cdot \size{\Dendo})$}
    and
    outputs a vector $q \in K$ that satisfies the following:
    \begin{align*}
        q(i, \true) = \sat_{Q, \Dexo, \Dendo}(i),& & \forall i \in \N
    \end{align*}
    \label{thm:correct:shapley}
\end{theorem}

\section{Correctness and Runtime Analysis of the Unifying Algorithm}
\label{sec:correctness}

\rev{In this section, we give a 
correctness proof and a runtime analysis of Algorithm~\ref{alg:hierarchical}, both of which apply
 regardless of the specific application and the specific \twomonoid that is being used.
We then show how to instantiate both the correctness proof and the runtime analysis
to each one of the three problems we consider in this paper.}

\subsection{General Correctness Proof}
\begin{definition}[Provenance Tree]
    \label{defn:provenance-tree}
    Let $\Sigma$ be a set of symbols.
    A {\em provenance tree} is a rooted tree whose leafs are labeled with symbols from $\Sigma \cup\{\true, \false\}$ and whose internal nodes are labeled with either $\wedge$ or $\vee$.
    A provenance tree is called {\em decomposable} if its leaves have distinct labels,
    \rev{including the labels $\true$ and $\false$.\footnote{\rev{The labels $\true$ and $\false$ do not need to appear in any trees except for the two trivial trees: the tree that consists of a single node which is $\true$, and the tree that consists of a single $\false$ node. This is because whenever $\true$ appears under a conjunction, it can be dropped, and whenever it appears under a disjunction, the entire disjunction becomes $\true$. Similar simplifications apply for $\false$, and these simplifications can always eliminate $\true$ and $\false$ from any node in a given tree, except for the root node.}}}
    Given a provenance tree $x$, the {\em support} of $x$, denoted $\supp(x)$, is the set of all symbols that appear in the leaves of $x$, excluding $\true$ and $\false$.
\end{definition}

\begin{definition}[Provenance \twomonoid]
    \label{defn:provenance-semiring}
    Let $\Sigma$ be a set of symbols.
    The {\em provenance \twomonoid} is a \twomonoid $\ov{\bm K} = (\ov{K}, \ov\oplus, \ov\otimes)$ defined as follows.
\begin{itemize}
    \item The domain $\ov{K}$ is the set of all provenance trees over $\Sigma\cup\{\true, \false\}$.
    \item The $\ov\oplus$ operator takes two provenance trees $x$ and $y$ and returns a new provenance tree $z$
    that has a root labeled with $\vee$ and whose children are the roots of $x$ and $y$.
    \item The $\ov\otimes$ operator takes two provenance trees $x$ and $y$ and returns a new provenance tree $z$
    that has a root labeled with $\wedge$ and whose children are the roots of $x$ and $y$.
\end{itemize}
In order to ensure commutativity of $\ov\oplus$ and $\ov\otimes$, we treat the children of any node
as an unordered set.
In order to ensure the associativity of $\ov\oplus$ and $\ov\otimes$, whenever we have
a node that has the same label as its parent, we merge them into a single node.
The identity for $\ov\oplus$ is the provenance tree with a single leaf labeled with $\false$,
and the identity for $\ov\otimes$ is the provenance tree with a single leaf labeled with $\true$.
\end{definition}

\begin{lemma}[Algorithm~\ref{alg:hierarchical} produces {\em decomposable} provenance trees]
    \label{lmm:decomposable}
    Given a hierarchical \sjfbcq $Q$ and a database instance $\ov\calD$,
    suppose that we annotate every fact in $\ov\calD$ with a {\em unique} symbol from a set of symbols $\Sigma$.
    Consider the provenance \twomonoid $\ov{\bm K} = (\ov{K}, \ov\oplus, \ov\otimes)$ from Definition~\ref{defn:provenance-semiring} where $\ov{K}$ is the set of provenance trees over $\Sigma\cup\{\true, \false\}$.
    The output of Algorithm~\ref{alg:hierarchical} on $Q$ and $\ov\calD$ is a provenance tree that is {\em decomposable} (Definition~\ref{defn:provenance-tree}).
    Moreover, every $\ov{\bm K}$-annotation that is computed throughout the algorithm is also a provenance tree that is {\em decomposable}.
\end{lemma}

\begin{proof}
    In order to prove the above lemma, we inductively prove that the algorithm maintains the following invariant. Recall the notion of {\em support} of a provenance tree from Definition~\ref{defn:provenance-tree}.
    \begin{claim}[Invariant of Algorithm~\ref{alg:hierarchical}]
        Every fact is annotated with a provenance tree that is decomposable.
        Moreover, all facts have provenance trees whose supports are pairwise disjoint.
    \end{claim}
    Initially, the above claim holds because every fact is annotated with a provenance tree consisting of a single leaf labeled with a unique symbol from $\Sigma$.
    We now prove that applying both Rule 1 and Rule 2 maintains the above invariant.

    For Rule 1 (line~\ref{alg:meta:rule1} of Algorithm~\ref{alg:hierarchical}), every fact $R'(\bm x')$ for $\bm x' \in \dom^{\bm X'}$ is annotated with a provenance tree which is the $\ov\oplus$ of provenance trees of all facts $R(\bm x', y)$ for $y \in \dom$.
    If all these provenance trees are decomposable and have disjoint supports, then the resulting
    provenance tree is also decomposable and has a support that is the union of the supports of all the children.
    Moreover, notice that the provenance tree of any fact $R(\bm x', y)$ is only used
    in the provenance tree of a single fact $R'(\bm x')$. Therefore, after we replace the atom $R(\bm X)$ with $R'(\bm X')$ (thus replacing all facts $R(\bm x', y)$ with new facts $R'(\bm x')$), the resulting facts will still have annotations whose supports are pairwise disjoint.

    For Rule 2 (line~\ref{alg:meta:rule2} of Algorithm~\ref{alg:hierarchical}), every fact $R'(\bm x)$ for $\bm x \in \dom^{\bm X}$ is annotated with a provenance tree which is the $\ov\otimes$ of provenance trees of the facts $R_1(\bm x)$ and $R_2(\bm x)$.
    If these two provenance trees are decomposable and have disjoint supports, then the resulting provenance tree is also decomposable and has a support that is the union of their supports.
    Moreover, every fact $R_1(\bm x)$ or $R_2(\bm x)$ is only used in the provenance tree of a single fact $R'(\bm x)$. Therefore, after we replace the two atoms $R_1(\bm X)$ and $R_2(\bm X)$ with $R'(\bm X)$ (thus replacing all facts $R_1(\bm x)$ and $R_2(\bm x)$ with new facts $R'(\bm x)$), the resulting facts will still have annotations whose supports are pairwise disjoint.
\end{proof}

The following theorem shows how to use the provenance \twomonoid from Definition~\ref{defn:provenance-semiring} as a {\em universal} \twomonoid to analyze the behavior of  Algorithm~\ref{alg:hierarchical} on any \twomonoid, and prove its correctness.

\begin{theorem}[Algorithm~\ref{alg:hierarchical} is correct]
    \label{thm:correct}
Let $Q$ be a hierarchical \sjfbcq and $\ov\calD$ be a database instance where every fact is annotated with a unique symbol from a set of symbols $\Sigma$.
Let $\ov{\bm K}= (\ov{K}, \ov{\oplus}, \ov{\otimes})$ be the provenance \twomonoid from Definition~\ref{defn:provenance-semiring} where $\ov{K}$ is the set of provenance trees over $\Sigma\cup\{\true, \false\}$,
and let $\ov{\bm 0}$ and $\ov{\bm 1}$ be the identities.
Let $\bm K = (K, \oplus, \otimes)$ be another \twomonoid with identities $\bm 0$ and $\bm 1$.
Suppose we are given a function $\phi: \ov{K} \to K$ that maps provenance trees to values in $K$ and satisfies the following properties, for all $x$ and $y$ in $\ov{K}$ that are decomposable and have disjoint supports:
\begin{align}
    \phi(x \ov{\oplus} y) &= \phi(x) \oplus \phi(y)\label{eq:correct:oplus}\\
    \phi(x \ov{\otimes} y) &= \phi(x) \otimes \phi(y)\label{eq:correct:otimes}\\
    \phi(\ov{\bm 0}) &= \bm 0
\end{align}
Let $\calD$ be the $\bm K$-annotated database instance obtained by replacing every $\ov{\bm K}$-annotation in $\ov{\calD}$ with its {\em $\phi$-mapping}, i.e. its corresponding value in $K$ using $\phi$.
Then, the output of Algorithm~\ref{alg:hierarchical} on $Q$ and $\calD$ is the $\phi$-mapping
of the output of Algorithm~\ref{alg:hierarchical} on $Q$ and $\ov\calD$.
\end{theorem}

Note that we do {\em not} require Eq.~\eqref{eq:correct:oplus} and Eq.~\eqref{eq:correct:otimes} to hold for {\em all} provenance trees in $\ov{K}$, but only for those that are decomposable and have disjoint supports. This is going to be crucial because in all applications we consider,
Eq.~\eqref{eq:correct:oplus} and Eq.~\eqref{eq:correct:otimes} are {\em only} going to hold for decomposable provenance trees with disjoint supports.

\begin{proof}
To prove the above theorem, we consider two simultaneous runs of Algorithm~\ref{alg:hierarchical}.
The first run is on $Q$ and $\calD$, while the second is on $Q$ and $\ov\calD$.
We show that they maintain the invariant:
\begin{claim}
Every annotation of a fact that is produced by the first run is the $\phi$-mapping of the annotation of the same fact that is produced by the second run.
\end{claim}
The above claim initially holds by the construction of $\calD$.
Thanks to Lemma~\ref{lmm:decomposable}, all provenance trees that are produced by the second run are decomposable.
Given two provenance trees $x$ and $y$, if $x \ov\oplus y$ (or $x \ov\otimes$ y) is decomposable, then $x$ and $y$ are also decomposable and have disjoint supports.
Hence, both Eq.~\eqref{eq:correct:oplus} and Eq.~\eqref{eq:correct:otimes} hold, thus they inductively prove the invariant.
\end{proof}


\subsection{General Runtime Analysis}

\begin{definition}[Size of a $\bm K$-annotated database]
Let $\bm K=(K, \oplus, \otimes)$ be a \twomonoid where $\oplus$ and $\otimes$ have identities $\bm 0$ and $\bm 1$ respectively. Given a $\bm K$-annotated relation $R(\bm X)$,
the {\em support} of $R(\bm X)$, denoted $\supp(R(\bm X))$, is the set of all facts $R(\bm x)$ in $R(\bm X)$ whose $\bm K$-annotation is not $\bm 0$.
The {\em size} of $R(\bm X)$ is the size of its support, i.e., $\card{\supp(R(\bm X))}$.
The support of a $\bm K$-annotated database instance $\calD$, denoted $\supp(\calD)$, is the union of the supports of its relations.
The {\em size} of a $\bm K$-annotated database instance $\calD$, denoted $\size{\calD}$, is the sum of the sizes of its relations.
\end{definition}

In order to analyze the runtime of Algorithm~\ref{alg:hierarchical}, we use the following lemma.

\begin{lemma}
Throughout Algorithm~\ref{alg:hierarchical}, the size of the $\bm K$-annotated database instance $\calD$ never increases.
\end{lemma}
\begin{proof}
    Algorithm~\ref{alg:hierarchical} updates the $\bm K$-annotated database instance $\calD$ in two places, namely lines~\ref{alg:meta:rule1:replace} and~\ref{alg:meta:rule2:replace}.
    We show below that neither update increases the size of $\calD$.
    \begin{itemize}
        \item In line~\ref{alg:meta:rule1:replace}, we replace the atom $R(\bm X', Y)$ with $R'(\bm X')$.
        We argue that the following holds: $\supp(R'(\bm X'))$ $\subseteq$ $\pi_{\bm X'}\supp(R(\bm X', Y))$.
        This is because for any given $\bm x' \in \dom^{\bm X'}$, we have
        $R'(\bm x') \defeq \bigoplus_{y \in \dom} R(\bm x', y)$
        and $\bm 0$ is the identity for $\oplus$.
        \item In line~\ref{alg:meta:rule2:replace}, we replace the two atoms $R_1(\bm X)$ and $R_2(\bm X)$ with a single atom $R'(\bm X)$.
        We argue that $\supp(R'(\bm X))\subseteq \supp(R_1(\bm X)) \cup \supp(R_2(\bm X))$,
        which implies $\card{\supp(R'(\bm X))} \leq \card{\supp(R_1(\bm X))}+\card{\supp(R_2(\bm X))}$.
        This holds because for any given $\bm x \in \dom^{\bm X}$, we have
        $R'(\bm x) = R_1(\bm x) \otimes R_2(\bm X)$, and a \twomonoid must satisfy $\bm 0 \otimes \bm 0 = \bm 0$ (Definition~\ref{defn:twomonoid}).
    \end{itemize}
\end{proof}

The following theorem is immediate.
\begin{theorem}[Algorithm~\ref{alg:hierarchical} used a linear number of $\oplus$ and $\otimes$ operations]
    \label{thm:runtime}
    Let $\bm K = (K, \oplus, \otimes)$ be a \twomonoid, $Q$ be a \sjfbcq, and $\calD$
    be a $\bm K$-annotated database instance of size $\size{\calD}$.
    On inputs $Q$ and $\calD$,
    the total number of $\oplus$ and $\otimes$ operations performed by Algorithm~\ref{alg:hierarchical} is $O(\size{\calD})$.
\end{theorem}

\rev{We are now ready to instantiate the correctness proof and the runtime analysis to each  of the three problems considered in this paper.
In particular, we will show how to apply Theorem~\ref{thm:correct} and Theorem~\ref{thm:runtime} to prove Theorems~\ref{thm:correct:probabistic},~\ref{thm:correct:repair}, and~\ref{thm:correct:shapley}.}
\subsection{First Instantiation: {\sc Probabilistic Query Evaluation}}

\begin{proof}[Proof of Theorem~\ref{thm:correct:probabistic}]
    First, we prove correctness using Theorem~\ref{thm:correct}.
    We define the function $\phi:\ov{K}\to K$ from Theorem~\ref{thm:correct} as follows: Given a provenance tree $x$, $\phi(x)$ is the probability of the Boolean formula 
    $F_x$ corresponding to the provenance tree $x$ evaluating to true.
    
    Given two decomposable provenance trees $x$ and $y$ with disjoint supports, consider the Boolean formulas $F_x$ and $F_y$ corresponding to $x$ and $y$ respectively.
    Because the supports are disjoint, the two formulas evaluating to true are independent events.
    Let $p_1$ and $p_2$ be the probabilities of $F_x$ and $F_y$, respectively, evaluating to true.
    By independence, the probability of $F_x \vee F_y$ is given by $p_1 \oplus p_2$ from Eq.~\eqref{eq:prob-semiring:oplus}, while the probability of $F_x \wedge F_y$ is given by $p_1 \otimes p_2$ from Eq.~\eqref{eq:prob-semiring:otimes}.
    Therefore, our definition of $\phi$ satisfies Eq.~\eqref{eq:correct:oplus} and Eq.~\eqref{eq:correct:otimes},
    and Theorem~\ref{thm:correct} applies.
    By the theorem, the output of Algorithm~\ref{alg:hierarchical} on $Q$ and $\calD$ is the probability of the Boolean formula corresponding to the provenance tree of $Q$ evaluating to true, as desired.

    Finally, the linear runtime follows from Theorem~\ref{thm:runtime} and the fact that the $\oplus$ and $\otimes$ operators from Eq.~\eqref{eq:prob-semiring:oplus} and Eq.~\eqref{eq:prob-semiring:otimes} take constant time.
\end{proof}

\subsection{Second Instantiation: \bagsetmax}

\begin{proof}[Proof of Theorem~\ref{thm:correct:repair}]
    First, we prove correctness using Theorem~\ref{thm:correct}.
    We define the function $\phi$ from Theorem~\ref{thm:correct} as follows:
    Given a provenance tree $x$, let $F_x$ be the corresponding Boolean formula. We define $\phi(x)$ as a monotonic vector in $\N^\N$ where for every $i\in\N$, the $i$-th entry is the maximum multiplicity of $F_x$ that can be achieved with a repair cost at most $i$, i.e.,
    by adding no more than $i$ facts from the repair database $\Dr$.

    Given two decomposable provenance trees $x$ and $y$ with disjoint supports, consider the Boolean formulas $F_x$ and $F_y$ corresponding to $x$ and $y$ respectively.
    Because the supports are disjoints, these two formulas have to be repaired independently.
    In particular, to repair the formula $F_x \vee F_y$ within a repair cost of $i$, we need
    to find the best way to repair $F_x$ within a repair cost of $i_1$ and $F_y$ within a repair cost of $i_2$ where $i_1 + i_2 = i$. Therefore, $\phi(x \ov\oplus y) = \phi(x) \oplus \phi(y)$
    where $\oplus$ is given by Eq.~\eqref{eq:repair-semiring:oplus}.
    The same reasoning goes for repairing the formula $F_x \wedge F_y$, leading to $\phi(x \ov\otimes y) = \phi(x) \otimes \phi(y)$
    where $\otimes$ is given by Eq.~\eqref{eq:repair-semiring:otimes}.
    Therefore, our definition of $\phi$ satisfies Eq.~\eqref{eq:correct:oplus} and Eq.~\eqref{eq:correct:otimes},
and Theorem~\ref{thm:correct} applies.
The theorem implies that the output of Algorithm~\ref{alg:hierarchical} on $Q$ and $\calD$ is a monotonic vector in $\N^\N$ whose $i$-th entry is the maximum multiplicity for $Q$ that can be achieved with a repair cost at most $i$, as desired.

Finally, we use Theorem~\ref{thm:runtime} to prove that the time complexity is $O((\size{\calD} +\size{\Dr})\cdot \size{\Dr}^2)$
\rev{and the space complexity is $O((\size{\calD} +\size{\Dr})\cdot \size{\Dr})$}.
Given an input instance $(\calD, \Dr, \theta)$ to the \bagsetmax problem from Definition~\ref{defn:bag-set:max}, note that the maximum repair cost $\theta$ is always upper bounded by the size of the repair database $\Dr$.
Let $\bm K = (K, \oplus, \otimes)$ be the \twomonoid  from Definition~\ref{defn:repair-semiring},
and $\bm x_1$ and $\bm x_2$ be two vectors in $K$.
Note that we only need to maintain the first $\theta + 1$ entries of $\bm x_1$ and $\bm x_2$.
Hence, the sizes of these vectors can be assumed to be linear in $\theta$, and by extension,
in $|\Dr|$.
As a result, the $\oplus$ and $\otimes$ operators from Eq.~\eqref{eq:repair-semiring:oplus} and Eq.~\eqref{eq:repair-semiring:otimes} take time $O(\size{\Dr}^2)$
\rev{and space $O(\size{\Dr})$.}
By Theorem~\ref{thm:runtime}, the total number of $\oplus$ and $\otimes$ operations performed by
Algorithm~\ref{alg:hierarchical} is $O(\size{\psi(\calD, \Dr)}) = O(\size{\calD} + \size{\Dr})$,
where $\psi(\calD, \Dr)$ is the $\bm K$-annotated database instance from Definition~\ref{defn:repair-annotated-database}. This gives the desired complexity.
\end{proof}

\subsection{Third Instantiation: {\sc Shapley Value Computation}}

\begin{proof}[Proof of Theorem~\ref{thm:correct:shapley}]
    First, we prove correctness.
In order to apply Theorem~\ref{thm:correct}, we define the function $\phi$ as follows.
For convenience, we extend Definition~\ref{defn:sharpsat} to define $\sat_{Q, \Dexo, \Dendo}(k, b)$ (with an extra Boolean parameter $b \in \B$) as the number of possible subsets $\calD'\subseteq \Dendo$
of size $|\calD'| = k$ such that $Q(\Dexo \cup \calD')$ is $b$.
This way, we can think of $\sat_{Q, \Dexo, \Dendo}(k, b)$ as a vector in $\N^{\N\times \B}$ that is indexed by $(k, b) \in \N\times \B$.
In addition, we also extend Definition~\ref{defn:sharpsat} from Boolean conjunctive queries $Q$ to arbitrary Boolean formulas $F$.
Given a provenance tree $x$, let $F_x$ be the corresponding Boolean formula,
and $\Dendo[F_x]$ be the set of facts in $\Dendo$ that appear in $F_x$.
We define $\phi(x)$ as follows:
\begin{align}
    \phi(x) \defeq \sat_{F_x, \Dexo, \Dendo[F_x]}
\end{align}

Let $x$ and $y$ be two decomposable provenance trees with disjoint supports,
$F_x$ and $F_y$ be the corresponding Boolean formulas, and $\Dendo[F_x]$ and $\Dendo[F_y]$ be the sets of facts in $\Dendo$ that appear in $F_x$ and $F_y$ respectively.
Note that $\Dendo[F_x]$ and $\Dendo[F_y]$ are disjoint, and $\Dendo[F_x \vee F_y]$
is their union.
As a result we have:
\begin{align*}
    \sat_{F_x \vee F_y, \Dexo, \Dendo[F_x \vee F_y]}(k, b) =&\\
    \sum_{k_1, k_2 \in \N: k_1 + k_2 = k}\quad&\sum_{b_1, b_2 \in \B: b_1 \vee b_2 = b}
    \sat_{F_x, \Dexo, \Dendo[F_x]}(k_1, b_1) \times
    \sat_{F_y, \Dexo, \Dendo[F_y]}(k_2, b_2)\\
    \sat_{F_x \vee F_y, \Dexo, \Dendo[F_x \vee F_y]} \quad=&\quad
    \sat_{F_x, \Dexo, \Dendo[F_x]} \quad\oplus\quad \sat_{F_y, \Dexo, \Dendo[F_y]}\\
    \phi(x \ov\oplus y) \quad=&\quad \phi(x) \quad\oplus\quad \phi(y)
\end{align*}
where the $\oplus$ operator above is given by Eq.~\eqref{eq:shapley:oplus}.
Using a similar reasoning for $F_x\wedge F_y$, we have:
\begin{align*}
    \sat_{F_x \wedge F_y, \Dexo, \Dendo[F_x \wedge F_y]}(k, b) =&\\
    \sum_{k_1, k_2 \in \N: k_1 + k_2 = k}\quad&\sum_{b_1, b_2 \in \B: b_1 \wedge b_2 = b}
    \sat_{F_x, \Dexo, \Dendo[F_x]}(k_1, b_1) \times
    \sat_{F_y, \Dexo, \Dendo[F_y]}(k_2, b_2)\\
    \sat_{F_x \wedge F_y, \Dexo, \Dendo[F_x \wedge F_y]} \quad=&\quad
    \sat_{F_x, \Dexo, \Dendo[F_x]} \quad\otimes\quad \sat_{F_y, \Dexo, \Dendo[F_y]}\\
    \phi(x \ov\otimes y) \quad=&\quad \phi(x) \quad\otimes\quad \phi(y)
\end{align*}
Hence, Theorem~\ref{thm:correct} applies.

Finally, we prove that the time complexity is $O((\size{\Dexo} + \size{\Dendo})\cdot \size{\Dendo}^2)$ \rev{and the space complexity is $O((\size{\Dexo} + \size{\Dendo})\cdot \size{\Dendo})$.}
Note that by the definition of $\sat_{Q, \Dexo, \Dendo}(k, b)$, the value of $k$ cannot exceed
$\size{\Dendo}$.
Hence, the $\oplus$ and $\otimes$ operators from Eq.~\eqref{eq:shapley:oplus} and Eq.~\eqref{eq:shapley:otimes} run in time $O(\size{\Dendo}^2)$ \rev{and space $O(\size{\Dendo})$.}
By Theorem~\ref{thm:runtime}, the total number of $\oplus$ and $\otimes$ operations performed is $O(\size{\psi(\Dexo, \Dendo)}) = O(\size{\Dexo} + \size{\Dendo})$,
where $\psi(\Dexo, \Dendo)$ is the $\bm K$-annotated database instance from Definition~\ref{defn:shapley-annotated-database}.
\end{proof}

\section{Concluding Remarks}
\label{sec:conclusion}

In this paper, we established a dichotomy theorem for the computational complexity of the 
{\bagsetmax} problem, parameterized by the collection of self-join-free Boolean conjunctive queries. Furthermore, we showed that the boundary of this dichotomy is defined by the collection of 
hierarchical queries. More importantly, perhaps, we discovered a single unifying polynomial-time algorithm for hierarchical queries that applies not only to {\bagsetmax} but also to {\sc Probabilistic Query Evaluation} and to {\sc Shapley Value Computation}.

The work presented here motivates several different questions, including the following two.

\textbf{Question 1:}
Is there a dichotomy for {\bagsetmax} for the collection of \emph{all} conjunctive queries? If so, what is the boundary of this dichotomy?

Note that in the case of {\sc Probabilistic Query Evaluation},
Dalvi and Suciu \cite{DalviS07} established a dichotomy for the
collection of all conjunctive queries. In the case of  {\sc Shapley Value Computation}, however, no dichotomy theorem for the collection of all conjunctive queries is currently known.

\textbf{Question 2:} Are there other natural algorithmic problems in database theory that are tractable for hierarchical queries, and this tractability can be obtained via the unifying polynomial-time algorithm presented here?

\rev{Among the candidates are the problems of constant-delay enumeration of conjunctive query answers under updates, as well as counting answers under updates~\cite{10.1145/3034786.3034789}.
In both problems,} hierarchical queries once again show up in the dichotomy between tractable and intractable cases. It is not clear, however, how to view this tractability result as an instantiation of the unifying algorithm presented in this paper or what \twomonoid to use.
\rev{Recent work~\cite{lmcs:13059} generalizes this dichotomy for conjunctive queries under updates along two dimensions: (i) queries with {\em free access patterns}, and (ii) queries over {\em probabilistic databases}. It would be interesting to see if these generalizations can also be captured by our  framework.}
\begin{acks}
  This work was initiated while Abo Khamis, Kolaitis, Roy, and Tannen participated in
  the Fall 2023 program on {\em Logic and Algorithms in Databases and AI} at the Simons Institute for the Theory of Computing.
  Part of this work was done while Roy was a visiting scientist at RelationalAI.
\end{acks}

\bibliographystyle{ACM-Reference-Format}
\bibliography{bib}


\end{document}